\documentclass[10pt,twocolumn]{IEEEtran}
\usepackage{fontenc}
\usepackage{graphicx}
\usepackage[cmex10]{amsmath}
\usepackage{array}
\usepackage{amssymb}
\usepackage{amsthm}
\usepackage{xcolor}
\usepackage[FIGBOTCAP]{subfigure}
\usepackage{parskip}
\usepackage{mathrsfs}
\usepackage{bm}
\usepackage{booktabs}
\usepackage[noadjust]{cite}
\newtheorem{theorem}{Theorem}
\newtheorem{lemma}[theorem]{Lemma}
\newtheorem{corollary}[theorem]{Corollary}
\newtheorem{proposition}[theorem]{Proposition}
\theoremstyle{definition}

\renewcommand\arraystretch{1.2}
\newcommand{\argmax}{\operatornamewithlimits{argmax}}

\def\clap#1{\hbox to 0pt{\hss#1\hss}}

\begin{document}
\title{Locality statistics for anomaly detection in time series of graphs}
\author{Heng~Wang, Minh~Tang, Youngser~Park, and Carey~E.~Priebe*%
\thanks{Heng~Wang is with the Department of Applied Mathematics and
  Statistics, Johns Hopkins University, Baltimore, MD 21218
  USA. (email: hwang82@jhu.edu)}%
\thanks{Minh~Tang is with the Department of Applied Mathematics and
  Statistics, Johns Hopkins University, Baltimore, MD 21218
  USA. (email: mtang10@jhu.edu)}%
\thanks{Youngser~Park is with the Center of Imaging Science, Johns
  Hopkins University, Baltimore, MD 21218
  USA. (email:youngser@jhu.edu)}%
\thanks{Carey~E.~Priebe is with the Department of Applied Mathematics and
  Statistics, Johns Hopkins University, Baltimore, MD 21218
  USA. (phone: 410-516-7200; fax:410-516-7459; email: cep@jhu.edu)}}%
%\markboth{IEEE Transactions on Signal Processing}{}
\maketitle

% flatex input: [abstract.tex]
\begin{abstract}

The ability to detect change-points in a dynamic network or a time
series of graphs is an increasingly important task in many
applications of the emerging discipline of graph signal processing.
This paper formulates change-point detection as a hypothesis testing
problem in terms of a generative latent position model, focusing on
the special case of the Stochastic Block Model time series.  We
analyze two classes of scan statistics, based on distinct underlying
locality statistics presented in the literature.  Our main
contribution is the derivation of the limiting properties and power
characteristics of the competing scan statistics.  Performance is
compared theoretically, on synthetic data, and on the Enron email
corpus.  We demonstrate that both statistics are admissible in one
simple setting, while one of the statistics is inadmissible in a
second setting.
\end{abstract}


\ifCLASSOPTIONpeerreview
\begin{center} \bfseries EDICS Category: SSP-SSAN, SSP-DETC \end{center}
\fi
%\pagebreak
%\input{intro.tex}
% flatex input: [theory.tex]
\section{Introduction}
The change-point detection problem in a dynamic network is becoming
increasingly prevalent in many applications of the emerging discipline
of graph signal processing. Dynamic network data are often readily
observed, with vertices denoting entities and time evolving edges
signifying relationships between entities, and thus considered as a
time series of graphs which is a natural framework for investigation.
An anomalous signal is broadly interpreted as constituting a deviation
from some normal network pattern, e.g. a model-based characterization
such as large scan statistics (c.f. \S~\ref{sec:two-test-statistics})
or non-model based notions such as a community structure change, while
a change-point is the time-window during which the anomaly appears.

Recently, many tailor-made approaches based on different models,
aiming for change-point detection in graphs, have been proposed in a
growing literature. \cite{Heard10} designs a two-stage Bayesian
anomaly detection method for social dynamic graphs. Both its model and
parallelization in computation are built on the assumption that the
communication between each pair of individuals independently follows a
counting process.
% If the independence assumption cannot be relaxed
% while large volumes of data are to be processed, the increment of
% false alarm rate in first stage and computational burden without
% parallelization in Bayes learning might be problematic.
In \cite{Mongiovi13}, an algorithm called NetSpot is created to find
arbitrary but evolutionary anomalies that are maintained over a spatial
or time window, i.e., the anomalous signal does not appears and then
disappears instantaneously
% , this algorithm fails to capture momentary
% signal.
In \cite{Wang08}, the subgraph
anomaly detection problem in static graphs were analyzed through
likelihood ratio tests under a Poisson random graph model. Finally, in
\cite{Miller10}, the $L_1$ norm of the eigenvectors of the modularity
matrix were used for detection of an (anomalous) small dense subgraph
embedded inside a large, sparser graph.

In this paper, we approach the dynamic anomaly/change-point detection
problem through the use of locality-based scan statistics. Scan
statistics are commonly used in signal processing to detect a local
signal in an instantiation of some random field
\cite{Glaz01,kulldorff97}.  The idea is to scan over a small
time or spatial window of the data and calculate some locality
statistic for each window. The maximum of these locality statistics is
known as the scan statistic. Large values of the scan statistic
suggests existence of nonhomogeneity, for example, a local region with
significantly excessive communications. Under some homogeneity
hypothesis, change-point detection can then be reduced to statistical
hypotheses testing (c.f. \S~\ref{sec:change-point-detect}) using scan
statistics. For example, \cite{Castro13} builds up a simple testing
framework with the null hypothesis being Erd\"{o}s-R\'{e}nyi and the
alternative hypothesis being a graph containing an unusually dense
subgraph. In the static graph setting, detection boundaries and
conditions are given in \cite{Castro13} such that the scan statistics
they specified for the testing is non-negligibly powerful. To capture
anomalies (e.g. hacker attacks) in computer networks, \cite{Neil11}
employs scan statistics through two shapes of locality statistics:
'star' and 'k-path'. The power properties of 'star' as a locality
measure will be further explored in \S~\ref{sec:power-estimates-max-d}
here.

In this paper, we identify excessive communication activity in a
subregion of a dynamic network by employing the scan statistics
$S_{\tau,\ell,k}(t;\cdot)$ defined in
\S~\ref{sec:two-test-statistics}, with $\tau$ denoting the number of
vertex-standardization steps, $\ell$ denoting the number of
temporal-normalization steps and $k$ denoting local neighborhood
distance. We consider two variations of $S_{\tau,\ell,k}(t;\cdot)$,
namely $S_{\tau,\ell,k}(t;\Psi)$ and $S_{\tau, \ell, k}(t; \Phi)$
where $\Psi$ and $\Phi$ are two related but distinct locality
statistics. The use of the locality statistics $\Psi$ and $\Phi$ is
based upon earlier work of \cite{priebe04:_scan} and
\cite{wan07:_captur}. In particular, $\Psi$ is introduced in
\cite{priebe04:_scan} to detect the emergence of local excessive
activities in time series of Enron graphs whereas $\Phi$ is proposed
in \cite{wan07:_captur} to detect communication pattern changes in
their department email network. Using the locality statistic $\Psi$,
\cite{youngser13} constructs fusion statistics of graphs for anomaly
detection while \cite{priebe10:inference} presents an analysis of the
Enron data set to illustrate statistical inference for attributed
random graphs. However, all these cited works are mostly empirical in
nature and do not provide much theoretical analysis of these
locality-based scan statistics. Under the assumption that the time
series of graphs is stationary before a change-point, we demonstrate
in this paper that for $\tau = 1$ and $\ell = 0$, the limiting
$S_{\tau,\ell,k}(t;\Psi)$ and $S_{\tau,\ell,k}(t;\Phi)$ are the
maximum of random variables which, under proper normalizations, follow
a standard Gumbel $\mathcal{G}(0,1)$ distribution in the
limit. Through these limiting properties, comparative power analysis
between $S_{\tau,\ell,k}(t;\Psi)$ and $S_{\tau,\ell,k}(t;\Phi)$ for
$\tau = 1$ and $\ell = 0$ is performed We demonstrate that both $\Psi$
and $\Phi$ are admissible if $k=0$, while $\Psi$ is inadmissible if
$k=1$. We hope that these theoretical results will help motivate
subsequent work in understanding the interplay between locality
statistics, vertex and temporal normalizations, and inference in
time-series of graphs.

Our paper is structured as follows. We discuss a generative model for
time series of graphs in \S~\ref{sec:latent-position-model}. The
problem of change-point detection is formulated in
\S~\ref{sec:change-point-detect}. The formulation associates a
change-point in the time series with changes in the underlying
generative model. We introduce in \S~\ref{sec:two-test-statistics} two
closely related notions of locality statistic, $\Psi$ and $\Phi$, and
their corresponding scan statistics $S_{\tau,\ell,k}(t;\Psi)$ and
$S_{\tau,\ell,k}(t;\Phi)$. The limiting properties and
power characteristics for some representative instances of
$S_{\tau,\ell,k}(t;\cdot)$ are given in
\S~\ref{sec:power-estimates-max-d} and \ref{sec:power-estimates-scan},
while \S~\ref{sec:experiment} presents experimental results regarding
locality-based statistics on synthetic data and the Enron email
corpus. We conclude the paper in \S~\ref{sec:concl--disc} with
additional discussions of the two locality statistics and comments
about possible applications and extensions of the framework presented
herein.
\subsection{Notation}
We introduce some notation that will be used throughout this
paper. In this paper, we consider only undirected and unweighted
graphs without self-loops.  Generally, a graph is denoted by $G$, with
vertex set $V = V(G)$ and edge set $E = E(G)$. The number of vertices
of a graph is usually denoted by $n$. For a graph $G$ on $n$ vertices,
the vertex set is usually taken to correspond to the set $[n] =
\{1,2,\dots,n\}$. In our subsequent discussion, we might also
partition $V$ into subsets, or blocks. If $V$ is partitioned into $B$
blocks of size $n_1, n_2, \dots, n_B$ vertices, then, with a slight
abuse of notation, we shall denote by $[n_i]$ the vertices in block
$i$.

Let $G$ be a graph. For any $u,v \in V$, we write $u \sim v$ if there
exists an edge between $u$ and $v$ in $G$. We write $d(u,v)$ for the
shortest path distance between $u$ and $v$ in $G$. For $v \in V$, we
denote by $N_{k}[v;G]$ the set of vertices $u$ at distance at most $k$
from $v$, i.e., $N_{k}[v;G] = \{u \in V \colon d(u,v) \leq k \}$. For
$V' \subset V$, $\Omega(V',G)$ is the subgraph of $G$ induced by
$V'$. Thus, $\Omega(N_k[v;G],G)$ is the subgraph of $G$ induced by
vertices at distance at most $k$ from $v$.

\section{Random graph models}
\label{sec:latent-position-model}
In this section, we briefly
summarize the latent position model of \cite{hoff02:_laten}, the dot
product model of \cite{young07:_random} and stochastic blockmodel of
\cite{Holland1983} and \cite{Wang1987}, since they are the underlying generative
model for our graph $G_t$ at each time $t$.

The latent position model is motivated by the assumption that each
vertex $v$ is associated with a $K$ dimensional latent random vector
$X_v$. For any pair of vertices $u$ and $v$, conditioning on the two
latent positions $X_u$ and $X_v$, the existence of an edge between $u$
and $v$ is independently determined by a Bernoulli trial with
probability $f(X_u,X_v)$ where $f$ is a symmetric link function
$f:\mathbb{R}^{K}\times \mathbb{R}^{K}\to[0,1]$. Namely,
$\mathbf{1}_{\{u \sim v\}} \stackrel{ind}\sim Bernoulli(f(X_u,X_v))$.

The random dot product graph model (RDPM) \cite{young07:_random} is a
special case of the latent position model. In the random dot product
graph model, the link function $f$ is specified to be the Euclidean
inner product, i.e., $f(X_u,X_v)= \langle X_u, X_v \rangle$. Also, for
each vertex $v$, the latent random vector $X_v$ takes its values in
the unit simplex $\mathscr{S}$ so that $0\leq \langle X_u, X_v \rangle
\leq 1$ where
\begin{equation*}
  \label{eq:1} \mathscr{S}=\{x\in[0,1]^K: \sum_{k=1}^{K} x_k \leq 1\}.
\end{equation*}

The stochastic block model (SBM) of \cite{Holland1983} and \cite{Wang1987} is a
random graph model in which each vertex is randomly assigned a block
membership among $\{1,\dots,B\}$ where $B$ is the number of
blocks. Given block memberships, the connectivity probabilities among
all vertices are characterized by a $B\times B$ symmetric matrix
$\mathbf{P}$ where $\mathbf{P}_{j,k}$ denotes the block connectivity
probability between blocks $j$ and $k$. Namely, $\mathbf{1}_{\{u \sim
v\}} \stackrel{ind}\sim Bernoulli(\mathbf{P}_{j,k})$ given $u\in
[n_j]$ and $v\in [n_k]$.

In this paper, we shall assume that the time series
of random graphs $\{G_t\}$ are generated according to a stochastic
block model where the block membership of the vertices are fixed
across time while the connectivity probabilities matrix $\mathbf{P} =
\mathbf{P}_{t}$ may varies with time (c.f. our formulation of the
change-point detection problem in
\S~\ref{sec:change-point-detect}). That is to say, at some initial
time, say $t_0 = 0$, we randomly assign each vertex to a block membership among
$\{1,2,\dots,K\}$. Then at each subsequent time $t \geq t_0$, $G_{t}$ follows a SBM
with a $K \times K$ probability matrix $\mathbf{P}_t$, conditioned on
the initial block membership at time $t_0$.
Under this model, the graphs are
\emph{conditionally} independent over time, the conditioning being on
the block membership of the vertices. This assumption on the
generative model for the $\{G_t\}$ leads to a time series of graphs
where the graphs are ``weakly'' dependent, i.e., they are dependent
only on the block membership of the vertices at the initial time
$t_0$. If, instead, for each time $t$, we resample the vertices' block membership for
$G_t$ then the resulting time series of graphs is independent.

Our construction of a time series of graphs in terms of the SBM as
outlined above is a special case of the following model constructed
using the random dot product graphs\footnote{The Dirichlet
distribution is a multivariate generalization of the beta distribution
and correspond to a distribution of points in the unit simplex. The
Dirichlet distribution, $\mathrm{Dirichlet}(\vec{\alpha})$,
$\vec{\alpha}=(\alpha_1,\dots,\alpha_K), \alpha_j>0, 1\leq j\leq K$,
has density $f_{\vec{\alpha}}(x_1,\dots,x_K)=\dfrac{\Gamma(\sum_{j=1}^{K}\alpha_j)}{\prod_{j=1}^{K}\Gamma(\alpha_j)}\prod_{j=1}^{K}x_j^{\alpha_j-1},
0<x_j<1,\sum_{j=1}^{K}x_j=1.$} .
\begin{enumerate}
\item For each $v \in [n]$ and $t \in \mathbb{N}$,
\begin{equation*}
  X_{v}(t)\sim \mathrm{Dirichlet} (r_v\vec{\alpha}_v+\vec{1}).
\end{equation*}
\item
For each $t\in\mathbb{N}$ and pair of vertices $(u,v)$,
\begin{equation*}
P(u\sim v| \, \mathbf{X}(t)) =\langle X_{u}(t), X_{v}(t)\rangle.
\end{equation*}
\end{enumerate}
where $\vec{\alpha}_v\in \mathscr{S}$ is
a fixed location parameter for the Dirichlet distribution and $r_v$
is the concentration parameter that will be explained now.

It is worthwhile to note that $r_v=0$ for all $v\in[n]$ means all
vertices follow the same probabilistic behavior (uniform on the
simplex) and $\min r_v\to \infty$ implies that $X_v(t)$ has a point
mass distribution at $\vec{\alpha}_v$ for each vertex. In the case
$\min r_v\to \infty$, the random dot product model can be further
reduced to the stochastic block model (SBM) by letting vertices
sharing the same $\vec{\alpha}_v$ share the same block
membership. Next, we re-denote by $\vec{\alpha}_i$ the common
Dirichlet location parameter corresponding to block $[n_i]$ and $V$ is
partitioned into $B$ distinct blocks $[n_1],\dots,[n_B]$ if there are
$B$ distinct $\vec{\alpha}_i$'s in total. Accordingly, as $\min r_i\to
\infty$, $P(u\sim v| \, u\in[n_j], v\in [n_k]) \to\langle
\vec{\alpha}_j, \vec{\alpha}_k\rangle$. We note that the above
Dirichlet can be viewed as generating a time-series of graphs where
the graphs are also ``weakly'' dependent, e.g., dependency between
graphs at time $t$ and $t'$ being on the location and concentration
parameters $\{(\alpha_{v}, r_v)\}$ for the vertices. Other
generalizations of the above construction for generating time series
of graphs are also possible. See, e.g., \cite{lee11} and \cite{lee13} for
examples of constructions where the time series of graphs depends on
some underlying latent stochastic processes.

\section{Change-point detection problem in Stochastic Block model formulation}
\label{sec:change-point-detect} An important inference task in time
series analysis is the problem of anomaly or change-point
detection. An anomaly is broadly interpreted to mean deviation from a
``normal'' pattern and a change-point is the time-window during which
the anomalous deviation occurs.  For example, in social networks, we
usually represent a time-evolving collection of emails, phone calls,
web pages visits, etc.\ as a time series of graphs $\{G_t\}$ and we
want to infer, from $\{G_t\}$, if there exists anomalous activities,
e.g., excessive phone calls among a subgroup in the
network. In the detection problem described below in
\S~\ref{sec:change-point-detect} and its theoretical analysis
presented in \S~\ref{sec:power-estimates-max-d} and
\S~\ref{sec:power-estimates-scan}, we shall implicitly assume, for
ease of exposition, that the $\{G_t\}$ are independent. As we pointed
out in our discussion of the generative model for time-series of
graphs in \S~\ref{sec:latent-position-model}, this independence
corresponds to conditioning on the right parameters. In the setup of
our theoretical analysis in this paper, this corresponds to
conditioning on the block membership of the vertices, which are fixed
in time. Related discussions in the context of the latent process
models of \cite{lee11} and \cite{lee13} are given in \S~\ref{sec:concl--disc}.

Statistically speaking, we want to test, for an unknown
but non-random $t \in \mathbb{N}$, the null hypothesis $H_0$ that $t$
is not a change-point against the alternative hypothesis $H_A$ that
$t$ is a change-point. There are many different ways to formulate the
notion that $t$ is a change-point. The following formulation, in the
contex of our discussion, is
reasonable and sufficiently general and forms the basis of our
subsequent investigation.

We say that $t^{*}$ is a change-point for $\{G_t\}$ if there exists
distinct choices of $\mathbf{P}^{0}$, $\mathbf{P}^{A}$ independent of
$t$ such that
\begin{equation*}
  \label{eq:2}
  H_A:  G_t \sim
  \begin{cases}
    \mathrm{SBM}(\mathbf{P}^{0}, \{[n_i]\}) & \text{for $t\leq t^\ast-1$}    \\
    \mathrm{SBM}(\mathbf{P}^{A}, \{[n_i]\}) & \text{for $t\geq t^\ast$}    \\
  \end{cases},
\end{equation*}
where $\mathrm{SBM}(\mathbf{P}, \{[n_i]\})$ denote
the stochastic blockmodel with block connectivity probabilities
$\mathbf{P}$ and unknown, but fixed in time, block memberships $\{[n_i]\}$.
In contrast, the null hypothesis, i.e. the nonexistence of
change-point, is
\begin{equation*}
  \label{eq:2}
  H_0:  G_t \sim \mathrm{SBM}(\mathbf{P}^{0}, \{[n_i]\}) ~\text{for all $t$}.   \\
\end{equation*}
That is to say, under the alternative, at time
$t^{*}$, a subset of the vertices change their behavior. The vertices
whose behaviour changes correspond to the vertices with block
memberships whose corresponding rows in the connectivity matrix
changes, i.e., from $\mathbf{P}^{0}$ to $\mathbf{P}^{A}$.
As permutation
of the vertex block labels does not affect our subsequent analysis, we
will refer to $(t^{*}, \{[n_i]\}, \mathbf{P}^{0}, \mathbf{P}^{A})$ as
the change parameters. As a convention, if $t^\ast=\infty$, we assume
all vertices follow their original dynamics for all $t$.

In the following, we discuss a specific
form for $\mathbf{P}^{0}$ and $\mathbf{P}^{A}$, illustrating,
albeit in an exaggerated manner, the chatter anomaly, i.e., a
subset of vertices with altered communication
behavior in an otherwise stationary setting.

\begin{equation}
  \label{eq:9}
  \mathbf{P}^{0} =
  \left(
    \begin{array}{ccccc}
      p_1 & p_{1,2} & \dots & \dots &p_{1,B} \\
      p_{2,1} & h_2 & \ddots & & \vdots \\
      \vdots & \ddots & \ddots & \ddots & \vdots \\
      \vdots &  & \ddots & h_{B-1} & p_{B-1,B} \\
      p_{B,1} & \dots & \dots & p_{B,B-1} & p_B \\
    \end{array}
  \right),
\end{equation}
\begin{equation}
  \label{eq:10}
  \mathbf{P}^{A} =\left(
    \begin{array}{ccccc}
      p_1 & p_{1,2} & \dots & \dots &p_{1,B} \\
      p_{2,1} & h_2 & \ddots & & \vdots \\
      \vdots & \ddots & \ddots & \ddots & \vdots \\
      \vdots &  & \ddots & h_{B-1} & p_{B-1,B} \\
      p_{B,1} & \dots & \dots & p_{B,B-1} & p_B+\delta \\
    \end{array}
  \right),
\end{equation}
for some $\delta > 0$, with $n_1, n_2, \dots, n_B$ being of size
\begin{equation*}
    (n_1,n_2,\dots,n_B)=(\Theta(n),O(n),\dots,O(n)).
\end{equation*}
For this form of $\mathbf{P}^0$ and $\mathbf{P}^A$,
the blocks have their own (possibly distinct) self-connectivity
probabilities which are diagonal entries of matrices. In other words,
before the change-point, each of the blocks $i = 2$ up to $B-1$ have
self-connectivity probability $h_i$. The block $i=1$ is
of size $\Theta(n)$ with self-connectivity probability $p_1$,
representing the probabilistic behaviors of the vast majority of
actors in a very large network. The case where $h_2 >
p_1, \dots,h_{B-1} > p_1$ is of interest because we can consider each
of the $[n_i]$ as representing a ``chatty'' group for time $t \leq
t^\ast-1$, and at $t^\ast$, the previously non-chatty group $[n_B]$
becomes more chatty. See
\figurename~\ref{fig:notional_heterogeneous_null} for a notional
depiction of $\mathbf{P}^0$ and $\mathbf{P}^A$ for the case of $B =
3$ blocks with $p_1=p_3=p_{1,2}=p_{1,3}=p_{2,3}=p$. The detection of
this transition for the vertices in $[n_B]$ is one of the main reasons
behind the locality statistics that will be introduced in
\S~\ref{sec:two-test-statistics}.
\begin{figure}[htbp]
  \centering
  \includegraphics[width=8cm]{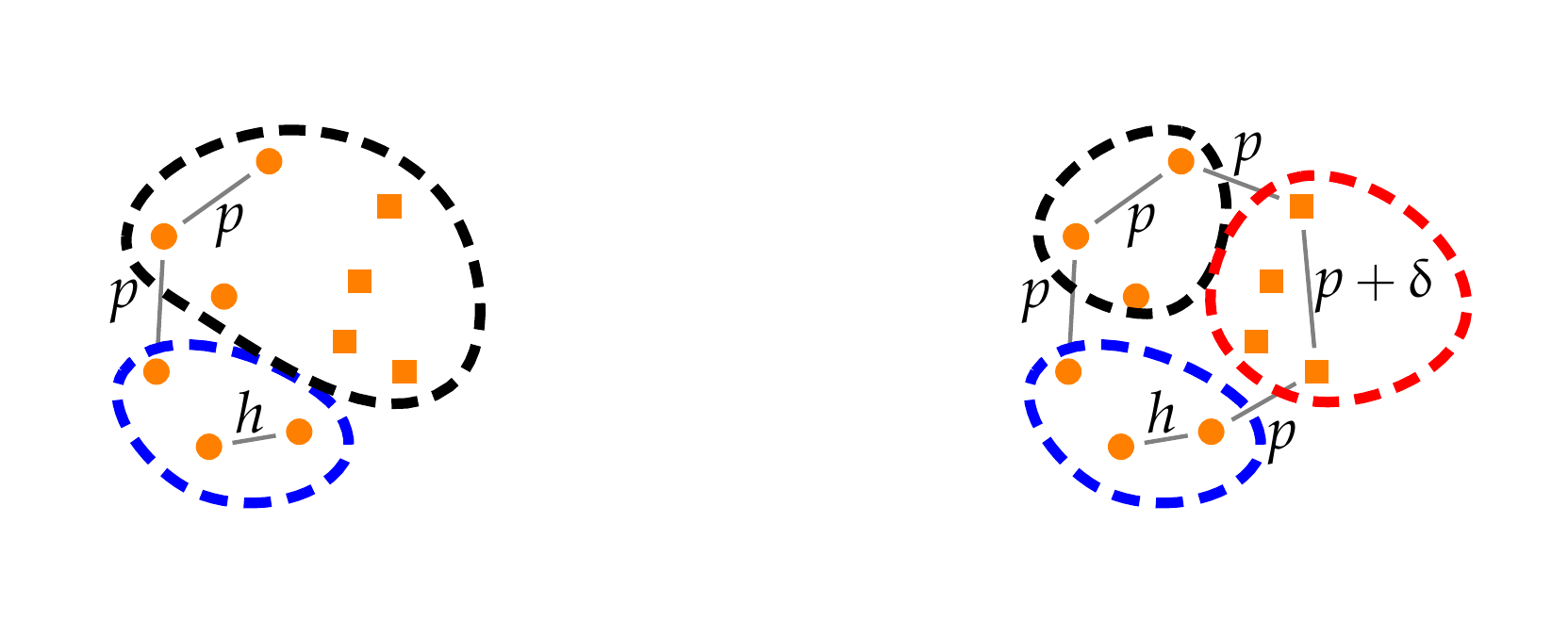}
  \caption{Notional depiction of $\mathbf{P}^0$ and corresponding
    $\mathbf{P}^A$. $\mathbf{P}^0$: all vertices connect with probability
    $p$ except that the self-connectivity probability of $[n_2]$ is $h$;
    $\mathbf{P}^A$: the self-connectivity probability of $[n_3]$
    transitions from $p$ to $p+\delta$ while $[n_2]$ retains its previous
    behavior.}
  \label{fig:notional_heterogeneous_null}
\end{figure}

\section{Locality statistics for change-point
 detection in time series of graphs}
\label{sec:two-test-statistics}
\subsection{Two locality statistics}
Suppose we are given a time series of graphs $\{G_t\}_{t \geq 1}$
where $V(G_t)$ is independent of $t$, i.e., the
graphs $G_t$ are constructed on the same vertex set $V$.
We now define two different but related locality statistics on
$\{G_t\}$. For a given $t$, let $\Psi_{t;k}(v)$ be defined for all
$k \geq 1$ and $v \in V$ by
\begin{equation}
  \label{eq:6}
 \Psi_{t;k}(v)=|E(\Omega(N_{k}(v;G_{t});G_{t}))|.
\end{equation}
$\Psi_{t;k}(v)$ counts the number of edges in the subgraph of $G_t$ induced
by $N_{k}(v;G_t)$, the set of vertices $u$ at a distance at most $k$ from $v$
in $G_t$. In a slight abuse of notation, we let
$\Psi_{t;0}(v)$ denote the degree of $v$ in $G_t$. The statistic
$\Psi_{t}$ was first introduced in
\cite{priebe04:_scan}. \cite{priebe05:_scan_statis_enron_graph}
investigated the use of $\Psi_{t}$ in analyzing the Enron data corpus.

Let $t$ and $t'$ be given, with $t' \leq t$. Now define
$\Phi_{t,t';k}(v)$ for all $k \geq 1$ and $v \in V$ by
\begin{equation}
  \label{eq:7}
\Phi_{t,t';k}(v)=|E(\Omega(N_{k}(v;G_{t});G_{t'}))|.
\end{equation}
The statistic $\Phi_{t,t';k}(v)$ counts the number of edges in the subgraph of
$G_{t'}$ induced by $N_{k}(v;G_t)$.

Once again, with a slight abuse of notation, we let $\Phi_{t,t';0}(v)$
denote the degree of $v$ in $G_{t} \cap G_{t'}$, where $G \cap G'$ for
$G$ and $G'$ with $V(G) = V(G')$ denotes the graph $(V(G), E(G) \cap
E(G'))$. The statistic $\Phi_{t,t';k}(v)$ is motivated by a statistic
named the permanent window metric introduced in
\cite{wan06:_statis}. The permanent window metric was meant to capture
events involving not just a single individual but the whole
community. As the community at time $t$ is assumed to be approximated
by $N_{k}(v;G_t)$, the statistic $\Phi_{t,t';k}(v)$ uses the community
structure at time $t$ in its computation of the locality statistic at
time $t' \leq t$.  Through this measure, a community structure shift
of $v$ can be captured even when the connectivity level of $v$ remains
unchanged across time, i.e., when the $\Psi_{t}$ stays mostly constant
as $t$ changes in some interval.  With the purpose of determining
whether $t$ is a change-point, two kinds of normalizations based on
past $\Psi$ and $\Phi$ locality statistics and their corresponding
normalized scan statistics are introduced in the next subsection.
\subsection{Temporally-normalized statistics}
Let $J_{t,t';k}$ be either the locality statistic $\Psi_{t';k}$ in
Eq.~(\ref{eq:6}) or $\Phi_{t,t';k}$ in Eq.~(\ref{eq:7}), where for
ease of exposition the index $t$ is a dummy index when $J_{t,t';k}
= \Psi_{t';k}$. We now define two normalized statistics for
$J_{t,t';k}$, a vertex-dependent normalization and a temporal
normalization. These normalizations and their use in the change-point
detection problem are depicted in \figurename~\ref{fig:temporal}.

\begin{figure}[htbp]
  \centering
  \includegraphics[width=9cm]{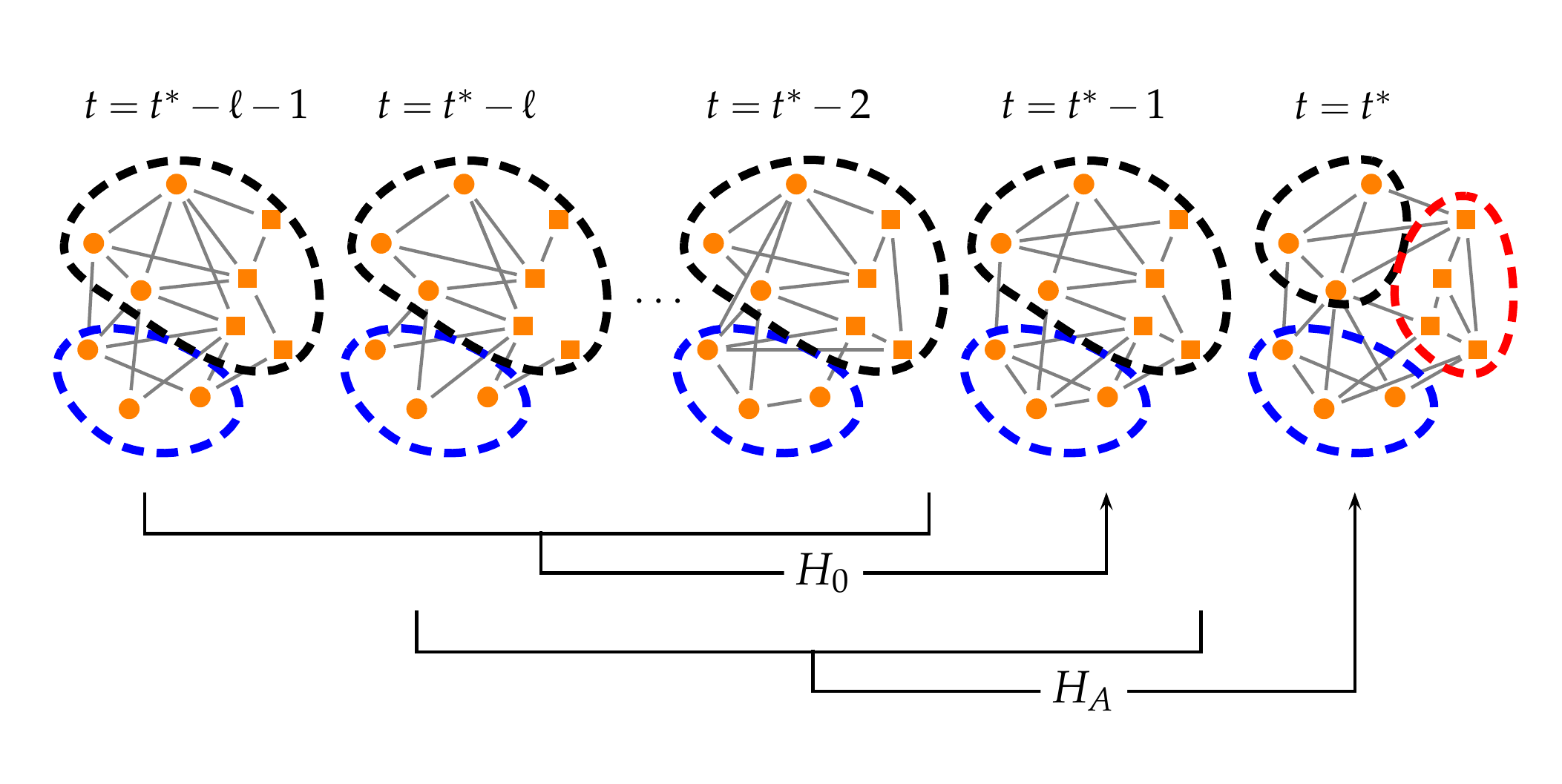}
  \caption{Temporal standardization: when testing for change at time
    $t$, the recent past graphs $G_t, G_{t-1}, \dots$ are used to
    standardize the invariants.}
  \label{fig:temporal}
\end{figure}

For a given integer $\tau \geq 0$ and $v \in V$, we define the
vertex-dependent normalization $\widetilde{J}_{t,
  \tau;k}(v)$ of $J_{t,t';k}(v)$ by
\begin{equation}\label{vertex-norm-eq1}
  \widetilde{J}_{t;\tau,k}(v)=
  \begin{cases}
    J_{t,t;k}(v) & \tau=0 \\
    J_{t,t;k}(v)- \hat{\mu}_{t;\tau,k}(v) & \tau = 1 \\
    (J_{t,t;k}(v) -
      \hat{\mu}_{t;\tau,k}(v))/\hat{\sigma}_{t;\tau,k} & \tau > 1
  \end{cases},
\end{equation}
where $\mu_{t;\tau,k}$ and $\sigma_{t;\tau,k}$ are defined as
\begin{gather}
  \hat{\mu}_{t;\tau,k}(v)=\frac{1}{\tau}\sum_{s=1}^{\tau}J_{t,t-s;k}(v), \label{vertex-norm-eq2}\\
  \hat{\sigma}_{t;\tau,k}(v)=\sqrt{\frac{1}{\tau-1}
\sum_{s=1}^{\tau}{(J_{t,t-s;k}(v)-\hat{\mu}_{t;\tau,k}(v))^2}}.
\end{gather}
We then consider the maximum of these vertex-dependent
normalizations for all $v \in V$, i.e., we define a
$M_{\tau,k}(t)$ by
\begin{equation}
  \label{eq:8}
M_{\tau,k}(t) =\max_v(\widetilde{J}_{t,\tau;k}(v)).
\end{equation}
We shall refer to $M_{\tau,0}(t)$ as the standardized
max-degree and to $M_{\tau,1}$ as the standardized scan
statistics. From Eq. (\ref{vertex-norm-eq1}), we see
that the motivation behind vertex-dependent normalization is to standardize
the scales of the raw locality statistics $J_{t,t';k}(v)$. Otherwise, in
Eq. (\ref{eq:8}), a noiseless vertex in the past who has dramatically
increasing communications at the current time $t$ would be
inconspicuous because there might exist a talkative vertex who keeps an
even higher but unchanged communication level throughout time.

Finally, for a
given integer $l \geq 0$, we define the temporal normalization of
$M_{\tau,k}(t)$ by
\begin{equation}\label{temporal-norm}
S_{\tau,\ell,k}(t)=
\begin{cases}
M_{\tau,k}(t) & \ell=0 \\
M_{\tau,k}(t) - \tilde{\mu}_{\tau,\ell,k}(t)  & \ell = 1 \\
(M_{\tau,k}(t) -
\tilde{\mu}_{\tau,\ell,k}(t))/\tilde{\sigma}_{\tau,\ell,k}(t)& \ell > 1
\end{cases},
\end{equation}
where $\tilde{\mu}_{\tau,\ell,k}$ and $\tilde{\sigma}_{\tau,\ell,k}$ are defined as
\begin{gather}
  \tilde{\mu}_{\tau,\ell,k}(t)=\frac{1}{\ell}\sum_{s=1}^{\ell}M_{\tau,k}(t-s), \\
  \tilde{\sigma}_{\tau,\ell,k}(t)=\sqrt{\frac{1}{\ell-1}
\sum_{s=1}^{\ell}{(M_{\tau,k}(t-s)-\tilde{\mu}_{\tau,\ell,k}(t))^2}}.
\end{gather}
The motivation behind temporal
normalization, based on recent $\ell$ time steps, is to perform
smoothing for the statistics $M_{\tau,k}$, similar to how smoothing
is performed in time series analysis. Large values of the smoothed
statistic indicates an anomaly where there is an excessive
increase in communications among a subset of vertices. We will use
these $S_{\tau,\ell,k}$ as the test statistics for the change-point
detection problem described in \S~\ref{sec:change-point-detect}.

We note that because $\Psi_{t;k}(v)=\Phi_{t,t;k}(v)$ for
$M_{\tau,k}$ when $\tau = 0$, the choice of locality statistic for
$J_{t,t';k}$ does not matter when $\tau=0$. For convenience of notation, since
$S_{\tau,\ell,k}(t)$ is essentially a function of the $J_{t,t';k}$,
we denote by $S_{\tau,\ell,k}(t;\Psi)$
and $S_{\tau,\ell,k}(t;\Phi)$ the $S_{\tau,\ell,k}(t)$ when
the underlying statistic $J_{t,t';k}$ is $\Psi_{t';k}$ and
$\Phi_{t,t';k}$, respectively.

After the above introduction of the temporally-normalized statistics
$S_{\tau,\ell,k}(t;\cdot)$ with three parameters $\tau,\ell,k$ , we
now present a simple toy example to illustrate a key step in the
calculation of $S_{\tau,\ell,k}(t;\cdot)$, namely the calculation of
the vertex-dependent normalization $\widetilde{J}_{t;\tau,k}(v)$
presented in Eq. (\ref{vertex-norm-eq1}). In
\figurename~\ref{fig:S-cal-toyexample}, the table calculates
$\widetilde{J}_{t^\ast;\tau,k}(v)$, when $\tau=1$ and $v=e$, for
different underlying statistics $J_{t,t';k}$ and different values of
$k$. More concretely, because $\tau=1$,
$\widetilde{J}_{t^\ast;1,k}(e)=\Psi_{t^\ast;k}(e)-\Psi_{t^\ast-1;k}(e)$
if the underlying statistic is $\Psi_{t;k}(e)$ and
$\widetilde{J}_{t^\ast;1,k}(e)=\Phi_{t^\ast,
t^\ast;k}(e)-\Phi_{t^\ast, t^\ast-1;k}(e)$ if the underlying statistic
is $\Phi_{t,t';k}(e)$.

\begin{figure}[htbp]
  \centering
  \includegraphics[width=3.7in,angle=0]{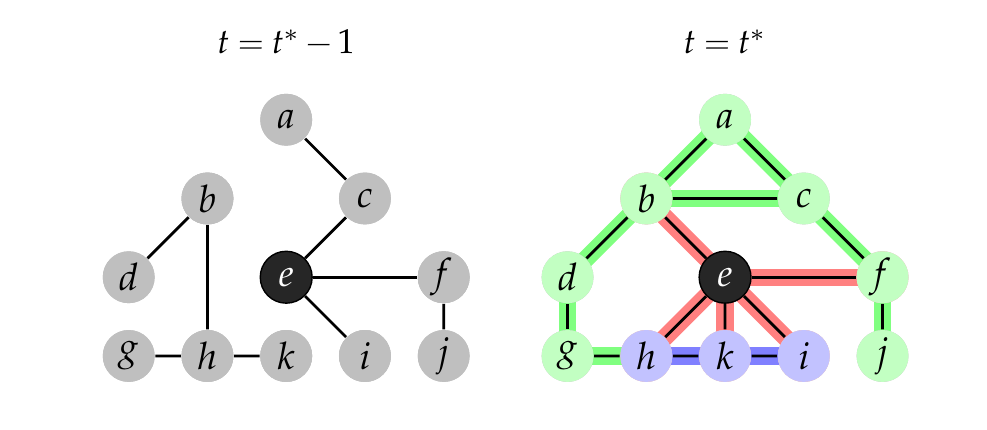}

  \begin{center}
    \begin{tabular}{rcccccc}
      \toprule
      &
      \multicolumn{2}{c}{$J^k_{t^*,t^*}(e)$}
      &
      \multicolumn{2}{c}{$\widehat{\mu}^k_{t^*,\tau}(e)$}
      &
      \multicolumn{2}{c}{$\widetilde{J}^k_{t^*,\tau}(e)$} \\
      \cmidrule(r){2-7}
      & $k=0$ & $k=1$ & $k=0$ & $k=1$ & $k=0$ & $k=1$ \\
      \midrule
      $\Psi_{t';k}$ & 5 & 7& {\color{magenta}3} & 3 & 2 & 4 \\
      $\Phi_{t,t';k}$ & 5 & 7& 2& {\color{orange}4} & 3 & 3 \\
      \bottomrule
      \label{tab:toy-example}
    \end{tabular}
  \end{center}

  \caption{An example to differentiate the calculation of
    $\widetilde{J}_{t^\ast;\tau,k}(v)$ with varying underlying statistics
    ($\Psi_{t;k}$ or $\Phi_{t,t';k}$) and order distances ($k=0$ or
    $k=1$). In the right graph $G_{t^\ast}$, note that the red edges are
    $E(\Omega(N_{k=0}[e;G_{t^\ast}],G_{t^\ast}))$; the red and blue edges
    are $E(\Omega(N_{k=1}[e;G_{t^\ast}],G_{t^\ast}))$; the red, blue and
    green edges are $E(\Omega(N_{k=2}[e;G_{t^\ast}],G_{t^\ast}))$. For
    instance, the magenta-marked number $3$ is $\Psi_{t^\ast-1;0}$ where
    $\Psi_{t^\ast-1;0}(e)=|E(\Omega(N_{0}(e;G_{t^\ast-1});G_{t^\ast-1}))|$
    and $E(\Omega(N_{0}(e;G_{t^\ast-1});G_{t^\ast-1}))=\{e\sim c, e\sim f,
    e\sim i\}$ in $G_{t^\ast-1}$.;the orange-marked number $4$ is
    $\Phi_{t^\ast,t^\ast-1;1}(e)$ where
    $\Phi_{t^\ast,t^\ast-1;1}(e)=|E(\Omega(N_{1}(e;G_{t^\ast});G_{t^\ast-1}))|$
    and $E(\Omega(N_{1}(e;G_{t^\ast});G_{t^\ast-1}))=\{h\sim k, b\sim h,
    e\sim i, e\sim f\}$ in $G_{t^\ast-1}$.}
  \label{fig:S-cal-toyexample}
\end{figure}

\section{Power Estimates of $S_{\tau=1,\ell=0,k=0}(t;\cdot)$}
\label{sec:power-estimates-max-d}
For algebraic simplicity, in Section
\ref{sec:power-estimates-max-d} and \ref{sec:power-estimates-scan}, we
consider a particularly simple form of $\mathbf{P}^0$ and
$\mathbf{P}^A$ where
\begin{equation}
  \label{eq:P0}
  \mathbf{P}^{0} =
  \left(
    \begin{array}{ccccc}
      p & p & \dots & \dots &p \\
      p & h_2 & \ddots & & \vdots \\
      \vdots & \ddots & \ddots & \ddots & \vdots \\
      \vdots &  & \ddots & h_{B-1} & p \\
      p & \dots & \dots & p & p \\
    \end{array}
  \right),
\end{equation}
\begin{equation}
  \label{eq:PA}
    \mathbf{P}^{A} =\left(
      \begin{array}{ccccc}
        p & p & \dots & \dots &p \\
        p & h_2 & \ddots & & \vdots \\
        \vdots & \ddots & \ddots & \ddots & \vdots \\
        \vdots &  & \ddots & h_{B-1} & p \\
        p & \dots & \dots & p & p+\delta \\
      \end{array}
    \right).
  \end{equation}
  With above form of $\mathbf{P}^0$ and $\mathbf{P}^A$,
in this section, we will derive the limiting properties
of $S_{1,0,0}(t;\Psi)$ and $S_{1,0,0}(t;\Phi)$ where
$S_{1,0,0}(t;\Psi)=\max_v(\Psi_{t;0}(v)-\Psi_{t-1;0}(v))$ and
$S_{1,0,0}(t;\Phi)=\max_v(\Phi_{t, t;0}(v)-\Phi_{t,
t-1;0}(v))$. Theorem \ref{thm:1} below shows that in the limit
$S_{\tau,\ell,k}(t;\cdot)$ is the maximum of random
variables that converge to the standard Gumbel distributions
$\mathcal{G}(0,1)$ under proper normalizations.
%To clarify the notation, for theorems and propositions in Section
% \ref{sec:power-estimates-max-d} and \ref{sec:power-estimates-scan},
% $S\stackrel{d}{\rightarrow}\sum_{i=1}^{B}\pi(n_i;\cdot)
% %\mathcal{G}(\mu(n_i;\cdot),\gamma(n_i;\cdot))$ expresses the meaning
% that there exists $Z\sim Multinomial(1,\vec{\pi})$ such that
% %$\dfrac{S-\mu(n_b;\cdot)}{\gamma(n_b;\cdot)}\stackrel{d}{\rightarrow}\mathcal{G}(0,1)$
% given $Z=b$. That is, the conditional distribution of $S|Z=b$ is
% $\mathcal{G}(\mu(n_b;\cdot),\gamma(n_b;\cdot))$ where $Z$ is a
% multinomial random variable with parameter $\vec{\pi}$.}

\begin{theorem}
  \label{thm:1}
  Let $\{G_t\}$ be a time series of random graphs according to the
  alternative $H_A$ detailed in
  \S~\ref{sec:change-point-detect}. In particular, $G_t\sim
  SBM(\mathbf{P}^0,\{[n_i]\}_{i=1}^{B})$ for $t \leq t^{\ast}-1$ and $G_t
  \sim SBM(\mathbf{P}^A, \{[n_i]\}_{i=1}^{B})$ for $ t \geq t^{\ast}$
  with $\mathbf{P}^0$ and $\mathbf{P}^A$ being of the form in
  Eq.~(\ref{eq:P0}) and Eq.~(\ref{eq:PA}), respectively. Let
  $S_{1,0,0}(t;\Psi)$ denote the statistic $S_{\tau,l,k}(t;\Psi)$ with
  $\tau = 1$, $l = 0$, and $k = 0$.
    Let $\mathcal{G}(\alpha,\gamma)$
  denote the Gumbel distribution with location parameter $\alpha$ and
  scale parameter $\gamma$. For a given $n \in \mathbb{N}$,
  let $a_n$ and $b_n$ be given by
\begin{gather*}
  \label{eq:13}
  a_n=\sqrt{2\log n}\Bigl(1-\frac{\log\log n+\log 4\pi}{4\log n}\Bigr), \\
  \label{eq:14}
  b_n=\frac{1}{\sqrt{2\log n}}.
\end{gather*}
Then as $n = \sum {n_i} \rightarrow \infty$,
$S_{1,0,0}(t;\Psi)$ has following properties:
\begin{gather}
  \label{eq:23}
S_{1,0,0}(t;\Psi) =\max_{1\leq i\leq B}W_0(n_i;\Psi) \quad t <
  t^{\ast}, \\
  \label{eq:24}
S_{1,0,0}(t;\Psi) =\max_{1\leq i\leq B}W_A(n_i;\Psi) \quad t = t^\ast,
\end{gather}

%i.e.,
%\begin{gather}
%  \label{eq:23}
%S_{1,0,0}(t;\Psi) \stackrel{d}\rightarrow \sum_{i=1}^{B}\pi_0(n_i;\Psi)
%\mathcal{G}(\mu_0(n_i;\Psi),\gamma_0(n_i;\Psi)) \quad t <
%  t^{\ast}, \\
%  \label{eq:24}
%S_{1,0,0}(t;\Psi) \stackrel{d}\rightarrow
%\sum_{i=1}^{B}\pi_A(n_i;\Psi)\mathcal{G}(\mu_A(n_i;\Psi),
%\gamma_A(n_i;\Psi))\quad t = t^\ast,
%\end{gather}
where
\begin{align*}
 % \label{eq:15}
  &\dfrac{W_0({n_i;\Psi)}-\mu_0(n_i;\Psi)}{\gamma_0(n_i;\Psi)}
    \stackrel{d}{\rightarrow}\mathcal{G}(0,1)\\
  &\dfrac{W_A({n_i;\Psi)}-\mu_A(n_i;\Psi)}{\gamma_A(n_i;\Psi)} \stackrel{d}{\rightarrow}\mathcal{G}(0,1)\\
\end{align*}
and the $\mu_0, \mu_A, \gamma_0, \gamma_A$ are given by
\begin{align*}
   & \mu_0(n_i;\Psi) = a_{n_i}\sqrt{Cnp(1-p)}  \\
   &\gamma_0(n_i;\Psi) = b_{n_i}\sqrt{Cnp(1-p)} \\
   &\mu_A(n_i;\Psi) = \mu_{0}(n_i;\Psi)+ \bm{1}_{\{i = B\}} n_B\delta\\
   &\gamma_A(n_i;\Psi) = \gamma_0(n_i;\Psi).
 \end{align*}
C is some explicit, computable constant.
%The $\pi_{\cdot}(n_i;\Psi)$ are {\color{red}multinomial} mixture coefficients defined by
%\begin{equation*}
%  \label{eq:17}
% \pi_{\cdot}(n_i;\Psi)
%= \mathbb{P}[X_i = \max \{X_1,X_2,\dots,X_B\}],
%\end{equation*}
%for $X_i \sim \mathcal{G}(\mu_{\cdot}(n_i;\Psi),\gamma_{\cdot}(n_i;\Psi))$.
Similarly, let $S_{1,0,0}(t;\Phi)$ denote $S_{\tau,l,k}(t;\Phi)$ with
$\tau = 1$, $l = 0$, and $k = 0$.
Then as $n = \sum {n_i} \rightarrow
\infty$,
\begin{gather}
  \label{eq:21}
S_{1,0,0}(t;\Phi)=\max_{1\leq i\leq B}W_0(n_i;\Phi) \quad t <
  t^{\ast}, \\
  \label{eq:22}
S_{1,0,0}(t;\Phi)=\max_{1\leq i\leq B}W_A(n_i;\Phi) \quad t = t^\ast,
\end{gather}
where
\begin{align*}
  \label{eq:16}
  &\dfrac{W_0({n_i;\Phi)}-\mu_0(n_i;\Phi)}{\gamma_0(n_i;\Phi)} \stackrel{d}{\rightarrow}\mathcal{G}(0,1)\\
  &\dfrac{W_A({n_i;\Phi)}-\mu_A(n_i;\Phi)}{\gamma_A(n_i;\Phi)} \stackrel{d}{\rightarrow}\mathcal{G}(0,1)\\
\end{align*}
and the $\mu_0, \mu_A, \kappa_0, \kappa_A$ in this case are
\begin{align*}
  &\kappa(p) = p(1-p)(1 - p(1-p)) \\
  &\xi_0(n_i;\Phi) = \bm{1}_{\{i \notin \{1,B\} \}}n_i(h_i(1 -
  h_i) - p(1-p)) \\
 &\mu_0(n_i;\Phi) = a_{n_i}\sqrt{Cn\kappa(p)} + np(1-p) + \xi_0(n_i;\Phi) \\
 &\gamma_0(n_i;\Phi) =b_{n_i}\sqrt{Cn\kappa(p)} \\
  &\mu_A(n_i;\Phi) = \mu_0(n_i;\Phi) + \bm{1}_{\{i = B\}} n_B \delta(1-p) \\
  &\gamma_A(n_i;\Phi) = \gamma_0(n_i;\Phi).
\end{align*}
\end{theorem}
We note the following corollary to Theorem~\ref{thm:1} for the case of
$B = 3$ blocks.

\begin{corollary} \label{corollary:max-d} Assume the setting in
Theorem~\ref{thm:1} with $B=3$. Let $\alpha > 0$ be given. Let
$\beta_\Phi$ be the power of the test statistic $S_{1,0,0}(t;\Phi)$
when $t = t^{\ast}$ for testing the hypothesis that $t$ is a change
point at a significance level of $\alpha$. Similarly, let
$\beta_\Psi$ be the power of the test statistic $S_{1,0,0}(t;\Psi)$
when $t = t^{\ast}$ for testing the same hypothesis at the same
significance level of $\alpha$. Then, as $(n_1,n_2,n_3)=(\Theta(n),O(n),O(n))$,
$\beta_\Phi, \beta_\Psi$ and $\alpha$ have the following relationship
% \footnote{ {\color{red} significance level $\alpha$ in Corollary
% \ref{corollary:max-d} and Proposition \ref{3block-power} represents
% the Type I error rate of the hypothesis testing. It is the probability
% of incorrectly rejecting the nonexistence of change-point. Besides,
% the meanings of all Big O notations in Corollary \ref{corollary:max-d}
% and Proposition \ref{3block-power} are detailed in
% \cite{big-oh-dict}}}:
\begin{enumerate}
\item $n_3=o{(\sqrt{n})}$ implies $\beta_\Phi=\alpha,\beta_\Psi=\alpha$.
\item $n_3=\Omega{(\sqrt{n})}$ implies $\beta_\Psi > \alpha$.
\item $n_3=\Theta{(\sqrt{n})} =\Theta(n_2)$ implies
  $\beta_\Phi >\alpha$.
\item $n_3=\omega{(\sqrt{n})} = \Theta{(n_2)}$ implies
  \begin{align*}
    \beta_\Phi &=\alpha & \text{if $\lim_{n\to\infty}\tfrac{n_2(h(1-h)-p(1-p))}{n_3\delta(1-p)}>1$}, \\
    \beta_\Phi & >\alpha & \text{if $\lim_{n\to\infty}\tfrac{n_2(h(1-h)-p(1-p))}{n_3\delta(1-p)}\leq1$}.
    \end{align*}
  \item $n_3=\Omega{(\sqrt{n})} =\omega(n_2)$ implies $\beta_\Phi
    > \alpha$.
  \item $n_3=\Omega{(\sqrt{n})} = o{(n_2)}$ implies
    \begin{align*}
      \beta_\Phi &=\alpha \,\, \text{if $h+p<1$}, \\
      \beta_\Phi &>\alpha  \,\, \text{if $h+p\geq 1$}.
    \end{align*}
  \end{enumerate}
\end{corollary}

From Corollary \ref{corollary:max-d}, an unanswered question is
whether there exists a dominance between $S_{1,0,0}(t;\Psi)$ and
$S_{1,0,0}(t;\Phi)$.  By using Theorem~\ref{thm:1}, we now present an
example to show that both statistics are admissible if we
restrict the test statistic space to only two
elements-$S_{1,0,0}(t;\Psi)$ and $S_{1,0,0}(t;\Phi)$. That is, neither
statistic has a statistical power dominance. Our setup is as
follows. Let $p = 0.43$. For each pair $(h,p+\delta)$ satisfying
$p<h<1$ and $p<p+\delta<1$, we generate a null and alternative
hypothesis pair $H_{0}$ and $H_{A}$ according to the model in
\S~\ref{sec:change-point-detect} with $B = 3$ blocks, i.e.,
\begin{equation*}
\mathbf{P}^0=
\left(
  \begin{array}{ccc}
    0.43 & 0.43  &0.43 \\
    0.43 & h  &0.43 \\
    0.43 & 0.43  &0.43 \\
  \end{array}
\right),
\mathbf{P}^A=\left(
  \begin{array}{ccc}
    0.43 & 0.43  &0.43 \\
    0.43 & h  & 0.43 \\
    0.43 & 0.43  &p+\delta \\
  \end{array}
\right).
\end{equation*}
with $n = n_1 + n_2 + n_3 = 1000$ and $n_1,n_2,n_3$ being functions of
$n, h$ and $\delta$ ($n_2 = n_3 = c_{p,h,\delta} \sqrt{n \log{n}}$
where the constant $c_{p,h,\delta}$ is dependent on $p,h$ and
$\delta$). In order to compare sensitivities of $S_{1,0,0}(t;\Psi)$
and $S_{1,0,0}(t;\Phi)$ in detection, we then calculate $\beta_{\Psi}-
\beta_{\Phi}$ by deriving the limiting property of
$S_{1,0,0}(t;\Psi)$ using Eqs.~(\ref{eq:23}) and (\ref{eq:24}) and the
limiting property of $S_{1,0,0}(t;\Phi)$ using
Eqs.~(\ref{eq:21}) and (\ref{eq:22}). The result is illustrated in
\figurename~\ref{fig:heatmap} where we have plotted $\beta_{\Psi} -
\beta_{\Phi}$ for different combinations of $h$ and $q
(=p+\delta)$. \figurename~\ref{fig:heatmap} indicates that the two
statistics $S_{1,0,0}(\cdot;\Psi)$ and $S_{1,0,0}(\cdot;\Phi)$ are
both admissible because $S_{1,0,0}(t;\Phi)$ achieves a
larger statistical power in the blue-colored region but a smaller power in
the red-colored region.
\begin{figure}[htbp]
  \centering
  \includegraphics[width=3.5in,angle=0]{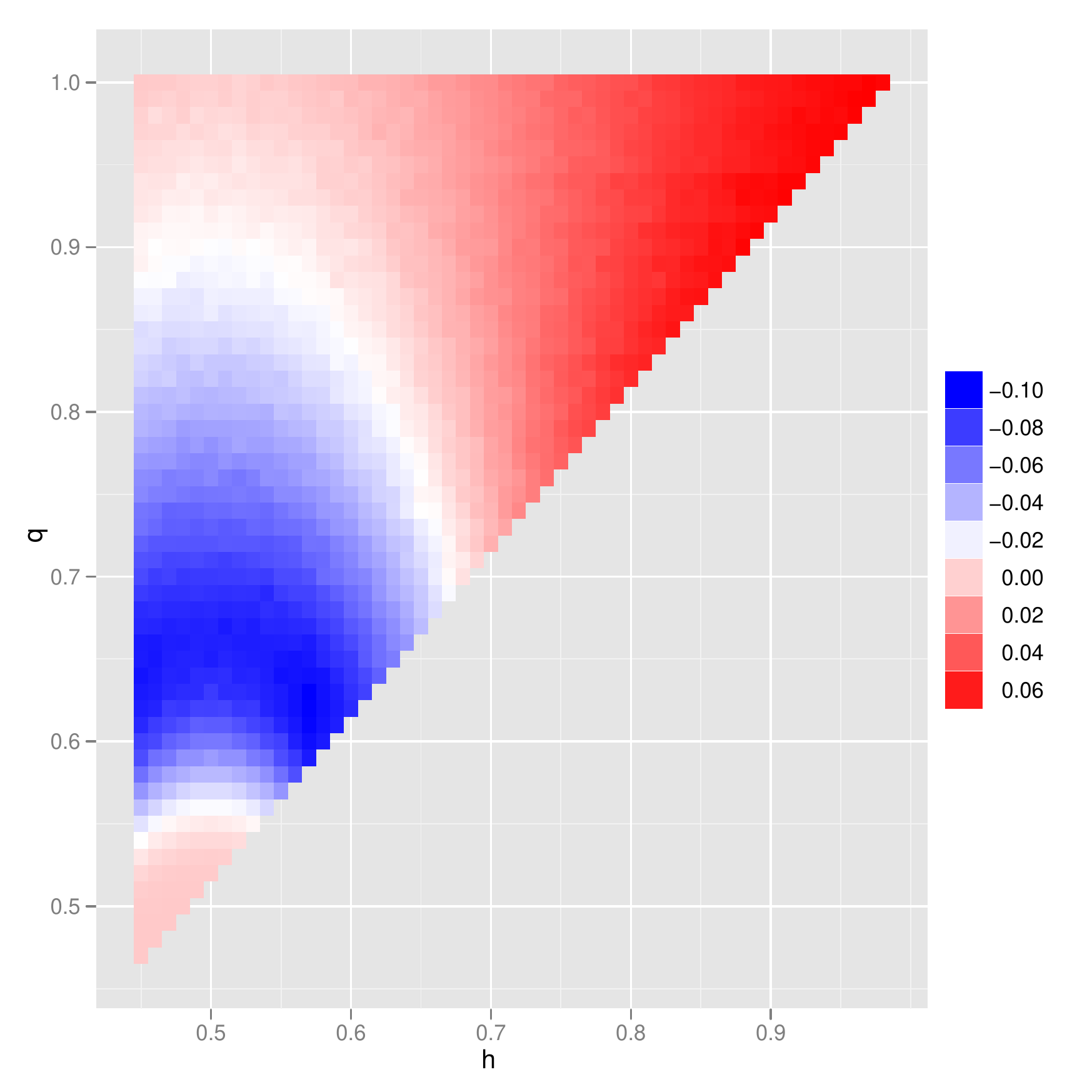}
  \caption{A comparison, using the limiting properties of
    $S_{1,0,0}(t;\Psi)$ and $S_{1,0,0}(t;\Phi)$, of $\beta_{\Psi} -
    \beta_{\Phi}$ for different null and alternative hypotheses
    pairs as parametrized by $h$ and $q (=p+\delta)$. The blue-colored region
    correspond to values of $h$ and $q (=p+\delta)$ for which $\beta_{\Psi} <
    \beta_{\Phi}$ while the red-colored region correspond to
    values of $h$ and $p+\delta$ with $\beta_{\Psi} > \beta_{\Phi}$.}
  \label{fig:heatmap}
\end{figure}

We now analyze the use of Theorem~\ref{thm:1} as a large-sample
approximation to $S_{1,0,0}(t;\Phi)$ and $S_{1,0,0}(t;\Psi)$. From
\figurename~\ref{fig:heatmap} with $p=0.43$, we choose a
$(h,p+\delta)$ pair, with $\beta_{\Psi} - \beta_{\Phi} > 0.05$, namely
$h = 0.95$ and $p+\delta = 0.98$. We then estimate the power of
$\beta_{\Phi}$ and $\beta_{\Psi}$ by repeated sampling of graphs from
stochastic blockmodel with parameters, $(\mathbf{P}^{0}, n_1, n_2,
n_3)$ for the null distribution and $(\mathbf{P}^{A}, n_1, n_2, n_3)$
for the alternative distribution. The result is presented in
\figurename~\ref{fig:power-maxd}. We see that the large-sample
approximation obtained via Theorem~\ref{thm:1} matches well with
sampling from the stochastic blockmodel
(SBM). \figurename~\ref{fig:power-maxd} also includes power estimates
for the random dot product model (RDPM) with varying concentration
parameter $r$ and predetermined location parameters
$\vec{\alpha}_1,\vec{\alpha}_2,\vec{\alpha}_3$.  Specifically,
$\vec{\alpha}_1,\vec{\alpha}_2,\vec{\alpha}_3$ are carefully chosen
such that their Euclidean inner products match corresponding block
connectivity probabilities i.e., $(p, h, q)$ specified above.  We see
that, as $r$ increases, the power estimates for the random dot product
model matches well with those of the stochastic blockmodel and
large-sample approximation. Finally \figurename~\ref{fig:power-maxd}
also includes power estimates for the locality statistics based on
$\Phi$ and $\Psi$ for $\tau = 0$, i.e., no vertex-dependent
normalization and is equivalent to the use of the max degree statistic
to test $H_0$ against $H_A$. These are represented as dashed and dot
blue lines, corresponding to large-sample approximation and Monte
Carlo simulations, respectively. Clearly, vertex-dependent
normalization leads to better performance for this $H_0$ and $H_A$
pair.
\begin{figure}
  \centering
  \includegraphics[width=3.5in,angle=0]{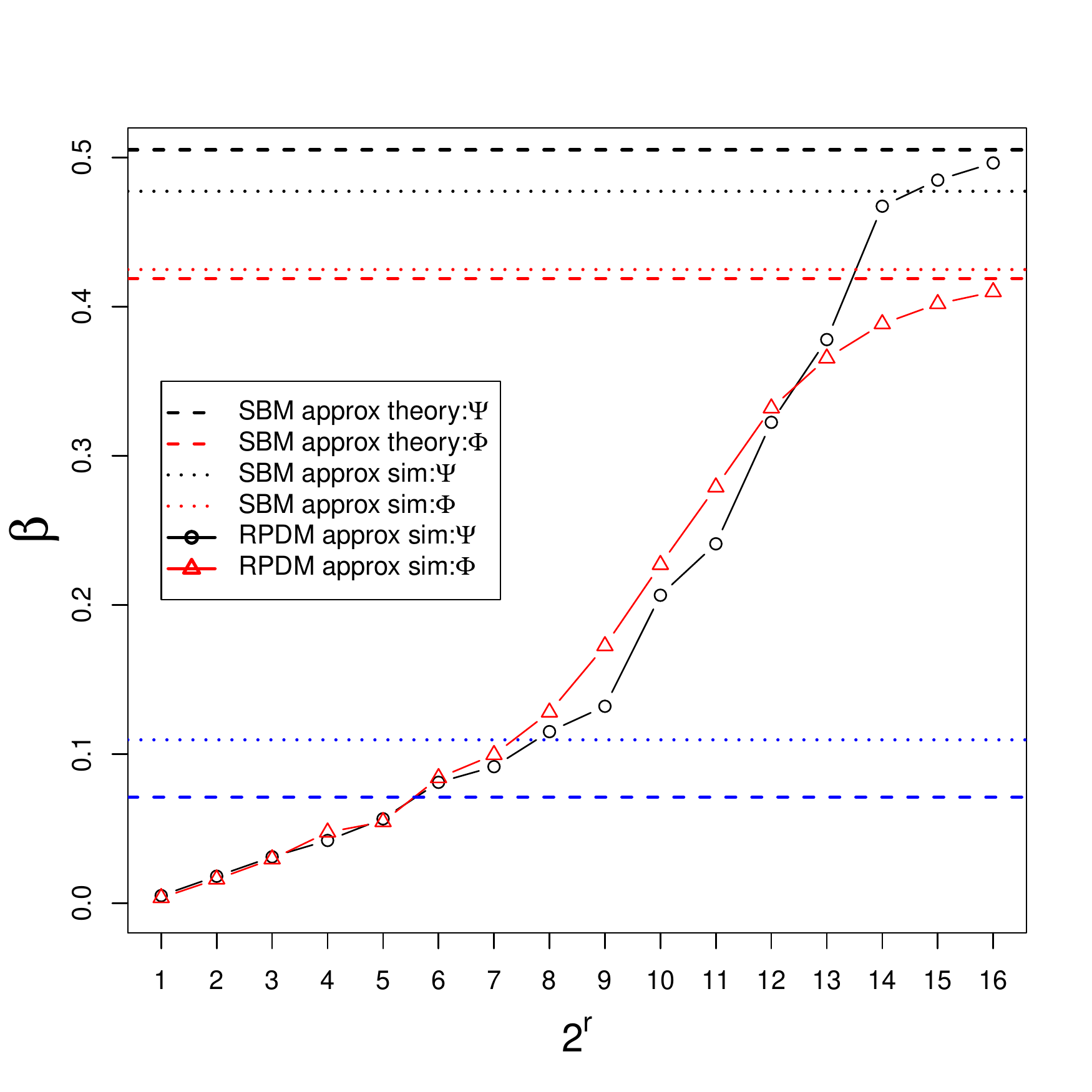}
  \caption{Power estimates $\beta_{\Psi}$ against $\beta_{\Phi}$ using
    Monte Carlo simulation on random graphs from the stochastic
    blockmodel, Monte Carlo simulation on random graphs from the random
    dot product model, and large-sample approximation for the stochastic
    blockmodel. $r$ is the concentration parameter. Dashed blue line:
    power estimate of large-sample approximation to $S_{0,0,0}(t;\Psi)$;
    dot blue line: power estimate of SBM Monte Carlo simulation to
    $S_{0,0,0}(t;\Psi)$.}
  \label{fig:power-maxd}
\end{figure}

\section{Power Estimates of $S_{\tau=1,\ell=0,k=1}(t;\cdot)$} \label{sec:power-estimates-scan}
In this section, we provide investigations of
$S_{\tau,\ell,k}(t;\Psi)$ and $S_{\tau,\ell,k}(t;\Phi)$ with a larger
scale parameter $k=1$ instead of $k=0$. We keep $\tau=1$
and $\ell=0$ the same as before and derive the limiting properties of
$\max_v(\Psi_{t;1}(v)-\Psi_{t-1;1}(v))$ and $\max_v(\Phi_{t,
t;1}(v)-\Phi_{t, t-1;1}(v))$. To make conclusions
concise and presentable, firstly, we delve into the
limiting properties in the model presented in \S~
\ref{sec:change-point-detect} with number of blocks $B=3$.
\begin{proposition} \label{3block-dist} Assume the same setting in
Theorem \ref{thm:1} with $B=3$. As
$(n_1,n_2,n_3)=(\Theta(n),o(n),o(n))$ and $n\to\infty$,
$S_{1,0,1}(t;\Psi)$ has the following properties:
\begin{gather*}
  \label{eq:33}
S_{1,0,1}(t;\Psi)=\max_{1\leq i\leq 3}W'_0(n_i;\Psi)  \quad t <
  t^{\ast}, \\
  \label{eq:34}
S_{1,0,1}(t;\Psi) =\max_{1\leq i\leq 3}W'_A(n_i;\Psi)  \quad t=t^\ast,
\end{gather*}
where
\begin{align*}
  &\dfrac{W'_0({n_i;\Psi)}-\mu'_0(n_i;\Psi)}{\gamma'_0(n_i;\Psi)} \stackrel{d}{\rightarrow}\mathcal{G}(0,1)\\
&\dfrac{W'_0({n_i;\Psi)}-\mu'_0(n_i;\Psi)}{\gamma'_0(n_i;\Psi)} \stackrel{d}{\rightarrow}\mathcal{G}(0,1)
  \end{align*}
  and the $\mu'_0, \mu'_A, \gamma'_0, \gamma'_A$ are given by
  \begin{align*}
   &\kappa'(n,p,n_2,h,i)=np^2+1+\bm{1}_{\{i=2\}}n_2p(h-p)\\
   &\mu'_0(n_i;\Psi) = \mu_0(n_i;\Psi)\kappa'(n,p,n_2,h,i)  \\
  &\gamma'_0(n_i;\Psi) = \gamma_0(n_i;\Psi)\kappa'(n,p,n_2,h,i)
 \end{align*}
  \begin{align*}
 \begin{split}
 \zeta(n_3,p,\delta,i)&=\dfrac{\delta}{2}[n_3^2(\bm{1}_{\{i\neq3\}}p^2+\bm{1}_{\{i=3\}}(p+\delta)^2)\\
 +&n_3(\bm{1}_{\{i\neq3\}}p(1-p)
 +\bm{1}_{\{i=3\}}(p+\delta)(1-p-\delta))]
  \end{split}
  \end{align*}
 \begin{align*}
 \begin{split}
 \mu'_A(n_i;\Psi)&=\mu_{A}(n_i;\Psi)[\kappa'(n,p,n_2,h,i)+\frac{\bm{1}_{\{i = 3\}}n_3p\delta}{2}]\\
                 &+\zeta(n_3,p,\delta,i)
  \end{split}
  \end{align*}
  \begin{align*}
 \gamma'_A(n_i;\Psi) &= \gamma_A(n_i;\Psi)[\kappa'(n,p,n_2,h,i)+\frac{\bm{1}_{\{i = 3\}}n_3p\delta}{2}].
\end{align*}

%The $\pi'_{\cdot}(n_i;\Psi)$ are {\color{red} multinomial} mixture coefficients defined by
%\begin{equation*}
%  \label{eq:27}
% \pi'_{\cdot}(n_i;\Psi)
%= \mathbb{P}[X_i = \max \{X_1,X_2,\dots,X_B\}],
%\end{equation*}
% for $X_i \sim \mathcal{G}(\mu'_{\cdot}(n_i;\Psi),\gamma'_{\cdot}(n_i;\Psi))$.
Likewise, \begin{align*}
S_{1,0,1}(t;\Phi) =\max_{1\leq i\leq 3}W'_0(n_i;\Phi)  ~ t<t^\ast,
\end{align*}
\begin{align*}
S_{1,0,1}(t;\Phi) =\max_{1\leq i\leq 3}W'_A(n_i;\Phi)  ~ t= t^\ast,
\end{align*}
where
\begin{align*}
&\dfrac{W'_0({n_i;\Phi)}-\mu'_0(n_i;\Phi)}{\gamma'_0(n_i;\Phi)} \stackrel{d}{\rightarrow}\mathcal{G}(0,1)\\
&\dfrac{W'_0({n_i;\Phi)}-\mu'_0(n_i;\Phi)}{\gamma'_0(n_i;\Phi)} \stackrel{d}{\rightarrow}\mathcal{G}(0,1)
\end{align*}
and the $\mu'_0, \mu'_A, \gamma'_0, \gamma'_A$ are given by
\begin{align*}
&\eta(p)=p^3(1-p)\\
&  \xi_0(n_i;\Phi) = \bm{1}_{\{i=2\}}n_2(h(1 -
  h) - p(1-p)) \\
&\mu'_0(n_i;\Phi)=a_{n_i}\sqrt{Cn^2\eta(p)}+np(1-p)+\xi_0(n_i;\Phi)\\
&\gamma'_0(n_i;\Phi)=b_{n_i}\sqrt{Cn^2\eta(p)}
\end{align*}
  \begin{align*}
 \begin{split}
   \zeta(n_3,p,\delta,i)&
   =\dfrac{\delta}{2}[n_3^2(\bm{1}_{\{i\neq3\}}p^2+\bm{1}_{\{i=3\}}(p+\delta)^2)\\
   &+n_3(\bm{1}_{\{i\neq3\}}p(1-p)+\bm{1}_{\{i=3\}}(p+\delta)(1-p-\delta))]
  \end{split}
  \end{align*}
\begin{align*}
  \mu'_A(n_i;\Phi)&=\mu'_0(n_i;\Phi)+\bm{1}_{\{i=3\}}n_3\delta(1-p)+\zeta(n_3,p,\delta,i)\\
  \gamma'_A(n_i;\Phi)&=\gamma'_0(n_i;\Phi)
\end{align*}
\end{proposition} Naturally, the limiting properties of
$S_{1,0,1}(t;\Psi)$ and $S_{1,0,1}(t;\Phi)$ as given above offer the
following power comparison result.
\begin{proposition}\label{3block-power}
  In the model shown in \figurename
\ref{fig:notional_heterogeneous_null}, Let $\alpha > 0$ be given,
$\beta'_\Phi$ be the power of the test statistic $S_{1,0,1}(t;\Phi)$
when $t = t^{\ast}$ for testing the hypothesis that $t$ is change
point at a significance level of $\alpha$ and $\beta'_\Psi$ be
the power of the test statistic $S_{1,0,1}(t;\Psi)$ when $t =
t^{\ast}$ for testing the same hypothesis at the same significance
level of $\alpha$.  As $n\to\infty$, $\beta'_\Phi, \beta'_\Psi$ and $\alpha$
have the following relationship:
    \begin{enumerate}
  \item $n_3=o(\sqrt{n})$ implies $\beta'_\Phi=\beta'_\Psi=\alpha$.
  \item $n_3=\Omega(\sqrt{n})$ implies
$\beta'_\Phi\geq\beta'_\Psi>\alpha$.
     \end{enumerate}
\end{proposition}
Consequently, Proposition \ref{3block-power} leads to the conclusion
that the performance of $S_{1,0,1}(t;\Phi)$ dominates
$S_{1,0,1}(t;\Psi)$ in the $3$-block model. Moreover, this superiority
can be generalized to the case with any given number of blocks
$B\geq3$. This is because each block $[n_i]$ with $1<i<B$ in
$B$-blocks model follows a similar probabilistic behavior as block
$[n_2]$ in $3$-blocks model while the power of hypothesis testing is
otherwise determined by the change of probabilistic behavior of block
$[n_B]$. In the limiting condition with $n\to\infty$, both
$\beta'_\Phi$ and $\beta'_\Psi$ in $B$-blocks model can be
characterized as a function of $p,\delta,n_B$ only. In other words,
though $h_2>p,\dots,h_{B-1}>p$, the "chatty" groups $[n_2],\dots,
[n_{B-1}]$ do not make any contribution on $\beta'_\Phi$ or
$\beta'_\Psi$. Hence, the number of "chatty groups", namely $B-2$, is
independent of the fact of dominance of $S_{1,0,1}(t;\Phi)$. Due to
the superiority of $S_{1,0,1}(t;\Phi)$, only the limiting
properties of $S_{1,0,1}(t;\Phi)$ in the general B-block model is
given below.

\begin{theorem}
  \label{thm:2}
  Let $\{G_t\}$ be a time series of random graphs according to the
  alternative $H_A$ detailed in
  \S~\ref{sec:change-point-detect}. In particular, $G_t\sim
  SBM(\mathbf{P}^0,\{[n_i]_{i=1}^{B}\})$ for $t < t^{\ast}$ and $G_t
  \sim SBM(\mathbf{P}^A, \{[n_i]_{i=1}^{B}\})$ for $ t \geq t^{\ast}$
  with $\mathbf{P}^0$ and $\mathbf{P}^A$ being of the form in
  Eq.~(\ref{eq:9}) and Eq.~(\ref{eq:10}), respectively. Let
  $S_{1,0,1}(t;\Phi)$ denote the statistic $S_{\tau,l,k}(t;\Phi)$ with
  $\tau = 1$, $l = 0$, and $k = 1$.

Then as $n = \sum {n_i} \rightarrow \infty$, $S_{1,0,1}(t;\Phi)$
has the following properties:
\begin{align*}
S_{1,0,1}(t;\Phi)=\max_{1\leq i\leq B}W'_0(n_i;\Phi)  \quad t<t^\ast,
\end{align*}
\begin{align*}
S_{1,0,1}(t;\Phi) =\max_{1\leq i\leq B}W'_A(n_i;\Phi) \quad t= t^\ast,
\end{align*}
where
\begin{align*}
&\dfrac{W'_0({n_i;\Phi)}-\mu'_0(n_i;\Phi)}{\gamma'_0(n_i;\Phi)} \stackrel{d}{\rightarrow}\mathcal{G}(0,1)\\
&\dfrac{W'_A({n_i;\Phi)}-\mu'_A(n_i;\Phi)}{\gamma'_A(n_i;\Phi)} \stackrel{d}{\rightarrow}\mathcal{G}(0,1)
\end{align*}
and the $\mu'_0, \mu'_A, \gamma'_0, \gamma'_A$ are given by
\begin{align*}
&\eta(p)=p^3(1-p)\\
&  \xi_0(n_i;\Phi) = \bm{1}_{\{i\notin\{1,B\}\}}n_i(h_i(1 -
  h_i) - p(1-p)) \\
&\mu'_0(n_i;\Phi)=a_{n_i}\sqrt{Cn^2\eta(p)}+np(1-p)+\xi_0(n_i;\Phi)\\
&\gamma'_0(n_i;\Phi)=b_{n_i}\sqrt{Cn^2\eta(p)}
\end{align*}
  \begin{align*}
 \begin{split}
 \zeta(n_B,p,\delta,i)&=\dfrac{\delta}{2}[n_B^2(\bm{1}_{\{i\neq B\}}p^2+\bm{1}_{\{i=B\}}(p+\delta)^2)\\
 +&n_B(\bm{1}_{\{i\neq B\}}p(1-p)+\bm{1}_{\{i=B\}}(p+\delta)(1-p-\delta))]
  \end{split}
  \end{align*}
\begin{align*}
\mu'_A(n_i;\Phi)&=\mu'_0(n_i;\Phi)+\bm{1}_{\{i=B\}}n_B\delta(1-p)+\zeta(n_B,p,\delta,i)\\
\gamma'_A(n_i;\Phi)&=\gamma'_0(n_i;\Phi)
\end{align*}
\end{theorem}

\begin{corollary} \label{corollary:scan}
Assume the setting in
Theorem~\ref{thm:2}. Let $\beta'_\Phi$ be the power of the test
statistic $S_{1,0,1}(t;\Phi)$ for $t = t^{\ast}$ and $\beta'_\Psi$ be
the power of the test statistic $S_{1,0,1}(t;\Psi)$ for $t =
t^{\ast}$. Then, as $(n_1,n_2,\dots,n_B)=(\Theta(n),o(n),\dots,o(n))$
and $n\to\infty$, $\beta'_\Phi \geq \beta'_\Psi$ and thus
$S_{1,0,1}(t;\Psi)$ is inadmissible.
\end{corollary}

In \S\ref{sec:power-estimates-max-d} and
\S\ref{sec:power-estimates-scan}, for simplicity of analytic
investigations, we theoretically obtain power estimates of
$S_{\tau,\ell,k}(t;\Psi)$ and $S_{\tau,\ell,k}(t;\Phi)$ under the
restrictions of $\tau=1$ and $\ell=0$. Besides analytic
investigations, we also empirically study power performances of
$S_{\tau,\ell,k}(t;\Psi)$ and $S_{\tau,\ell,k}(t;\Phi)$ with other
$(\tau,\ell)$ combinations via Monte Carlo simulations. In this
experiment, we let $\tau$ range from $0$ to $10$ and $\ell$ range from
$0$ to $10$. In each Monte Carlo replicate, a time series of random
graphs based on the SBM considered in
\S\ref{sec:power-estimates-max-d}, where
$(n_1,n_2,n_3)=(870,65,65),(p,h,q)=(0.43,0.95,0.98)$, is
sampled. Next,
$S_{\tau,\ell,k}(t^\ast-1;\Psi)$,$S_{\tau,\ell,k}(t^\ast-1;\Phi)$,$S_{\tau,\ell,k}(t^\ast;\Psi)$
and $S_{\tau,\ell,k}(t^\ast;\Phi)$ are calculated individually
according to specific $(\tau,\ell,k)$. After $2000$ replicates, for
each test statistic, the largest empirical power (denoted by
$\max_{(\tau,\ell)}\beta$) and the corresponding optimal choice
of $(\tau,\ell)$ (denoted by $(\tau^*,\ell^*)$) is obtained and
summarized in Table~\ref{optimal_tau-ell}.\\
\begin{table}[tbp]
  \centering
  \renewcommand{\arraystretch}{1.3}
  \caption{The optimal $\tau$ and $l$ in an experiment
    comparing the statistical power of $S_{\tau,\ell,k}$ for $k = 0,
    1$ and locality statistics $\Phi$ and $\Psi$. We varies $\tau,
    \ell \in \{0,1,\dots,10\}$ and compare the statistic power for
    each choice of $\tau$ and $\ell$ through a Monte-Carlo experiment
    with $2000$ replicates.}
\label{optimal_tau-ell}
\begin{tabular}{|c|c|c|}
  \hline
     & $\max_{(\tau,\ell)}\beta$ & $(\tau^*,\ell^*)$ \\
 $S_{\tau,\ell,0}(t;\Psi)$ & $0.483$ & $(1,0)$ \\
 $S_{\tau,\ell,0}(t;\Phi)$ & $0.384$ & $(1,10)$ \\
 $S_{\tau,\ell,1}(t;\Psi)$ & 0.571 & (1,10) \\
 $S_{\tau,\ell,1}(t;\Phi)$ & 0.758 & (1,9) \\
  \hline
\end{tabular}
\end{table}
The empirical results in Table~\ref{optimal_tau-ell} demonstrate the
potential value of extending the theoretical investigations in
\S\ref{sec:power-estimates-max-d} and \S\ref{sec:power-estimates-scan}
to cases of $\tau\geq 1$ and $\ell \geq 1$, though this extension
appears significantly more challenging than the case
$(\tau,\ell)=(1,0)$.
% flatex input end: [theory.tex]

%\input{intro.tex}
% flatex input: [enron_Heng.tex]
\section{Experiment}
\label{sec:experiment}
We use the Enron email data used in
\cite{priebe05:_scan_statis_enron_graph} for this experiment.
It consists of time series of graphs $\{G_t\}$ with $|V|=184$ vertices
for each week $t=1,\ldots,189$, where we
draw a unweighted edge when vertex $v$ sends at least one email to
vertex $w$ during a one week period.

After truncating first 40 weeks for vertex-standardized and temporal
normalizations, Figure \ref{fig:enron1} depicts
$S_{\tau,\ell,k}(t;\Psi)$ (SeaGreen) and $S_{\tau,\ell,k}(t;\Phi)$
(Orange) in the remaining 149 weeks from August 1999 to June 2002.
In this experiment, we choose both $\tau=\ell=20$, used in
\cite{priebe05:_scan_statis_enron_graph}, to keep the comparisons
between the two papers meaningful. As indicated in
\cite{priebe05:_scan_statis_enron_graph}, detections are defined as
weeks $t$ such that $S_{\tau,\ell,k} > 5$. Hence, from Figure
\ref{fig:enron1} we have following observations and reasonings.

\begin{enumerate}
\item \label{obs1}
$S_{20,20,0}(t;\Psi),S_{20,20,0}(t;\Phi),S_{20,20,1}(t;\Psi),S_{20,20,1}(t;\Phi)$
and $S_{20,20,2}(t;\Phi)$ indicate a clear anomaly at $t^*=58$ in
December 1999. This coincides with the happening of Enron's tentative
sham energy deal with Merrill Lynch to meet profit expectations and
boost stock price \cite{enron_timeline}. The center of suspicious
community-employee $v_{154}$ is identified by all five statistics.
\item \label{obs2}
$S_{20,20,0}(t;\Psi)$, $S_{20,20,0}(t;\Phi)$, $S_{20,20,1}(t;\Psi)$
and $S_{20,20,2}(t;\Psi)$ capture an anomaly at $t^\ast=146$ in the
mid-August 2001. This is the period that Enron CEO Skilling made a
resignation announcement when the company was surrounded by public
criticisms shown in \cite{enron_timeline}. The center of suspicious
community-employee $v_{95}$ is identified by these four statistics.
\item \label{obs3} $S_{20,20,2}(t;\Psi)$ signifies an anomaly at
$t^\ast=132$ in late April 2001 where $S_{20,20,k}(t;\Phi)$ fails to
alert for any $k\in\{0,1,2\}$. This phenomenon occurs because
$S_{20,20,2}(t;\Psi)$ captures the employee $v_{90}$ whose
second-order neighborhood $N_{2}(v_{90};G_{132})$ contains $116$
emails at $t^\ast=132$ but $0$ email in his second-order neighborhoods
of previous 20 weeks. That is, the time-dependent second-order
neighborhood $N_{2}(v_{90};G_{t})$ had no communication in the period
from $t=112$ to $t=131$. On the other hand, this behavior cannot be
monitored by $S_{20,20,2}(132;\Phi)$ because the change of
communication frequency in a fixed second-order neighborhood
$N_{2}(v_{90};G_{132})$, measured by locality statistics $\Phi$, is
not so significant. More concretely, the number of emails in the
unchanged $N_{2}(v_{90};G_{132})$ has a mean of $45.5$ and a standard
deviation of $14.9$ from $t=112$ to $t=131$. In \cite{enron_timeline},
this anomaly appears after the Enron Quaterly Conference Call in which
a Wall Street analyst Richard Grubman questioned Skilling on the
company's refusal of releasing balance sheet but then got insulted by
Skilling.
\item \label{obs4} $S_{20,20,2}(t;\Phi)$ shows a detection on
$v_{135}$ at $t^\ast=136$ before June 2001 over
$S_{20,20,2}(t;\Psi)$. This comes from the fact that the fixed
second-order neighborhood of employee $v_{135}$ at $t^\ast=136$,
i.e. $N_{2}(v_{135};G_{136})$, has a small standard deviation $1.08$
in previous 20 weeks while the communications in time-dependent
neighborhoods $\{N_{2}(v_{90};G_{t})\}_{t=116}^{135}$ has a large
standard deviation $10.04$. Practically speaking, in this case, a
dramatic increment of email contacts in the certain community
$N_{2}(v_{135};G_{136})$ could be captured by $S_{20,20,2}(t;\Phi)$
but ignored by $S_{20,20,2}(t;\Psi)$ because unstable communication
patterns in $\{N_{2}(v_{90};G_{t})\}_{t=116}^{135}$ offsets the
sensitivity of signal. According to \cite{enron_timeline}, this
anomaly corresponds to the formal notice of closure and termination of
Enron's single largest foreign investment, the Dabhol Power Company in
India.
\end{enumerate}

In summary, observations \ref{obs1} and \ref{obs2} demonstrate that in
some cases both $S_{\tau,\ell,k}(t;\Psi)$ and
$S_{\tau,\ell,k}(t;\Phi)$ are capable of capturing the same community
which has a significant increment of connectivity. Besides, in some
situations shown in observations \ref{obs3} and \ref{obs4},
$S_{\tau,\ell,k}(t;\Psi)$ and $S_{\tau,\ell,k}(t;\Phi)$ achieve
different detections due to its adaptability.

\begin{figure}[!t]
  \centering
 {\includegraphics[width=3.6in]{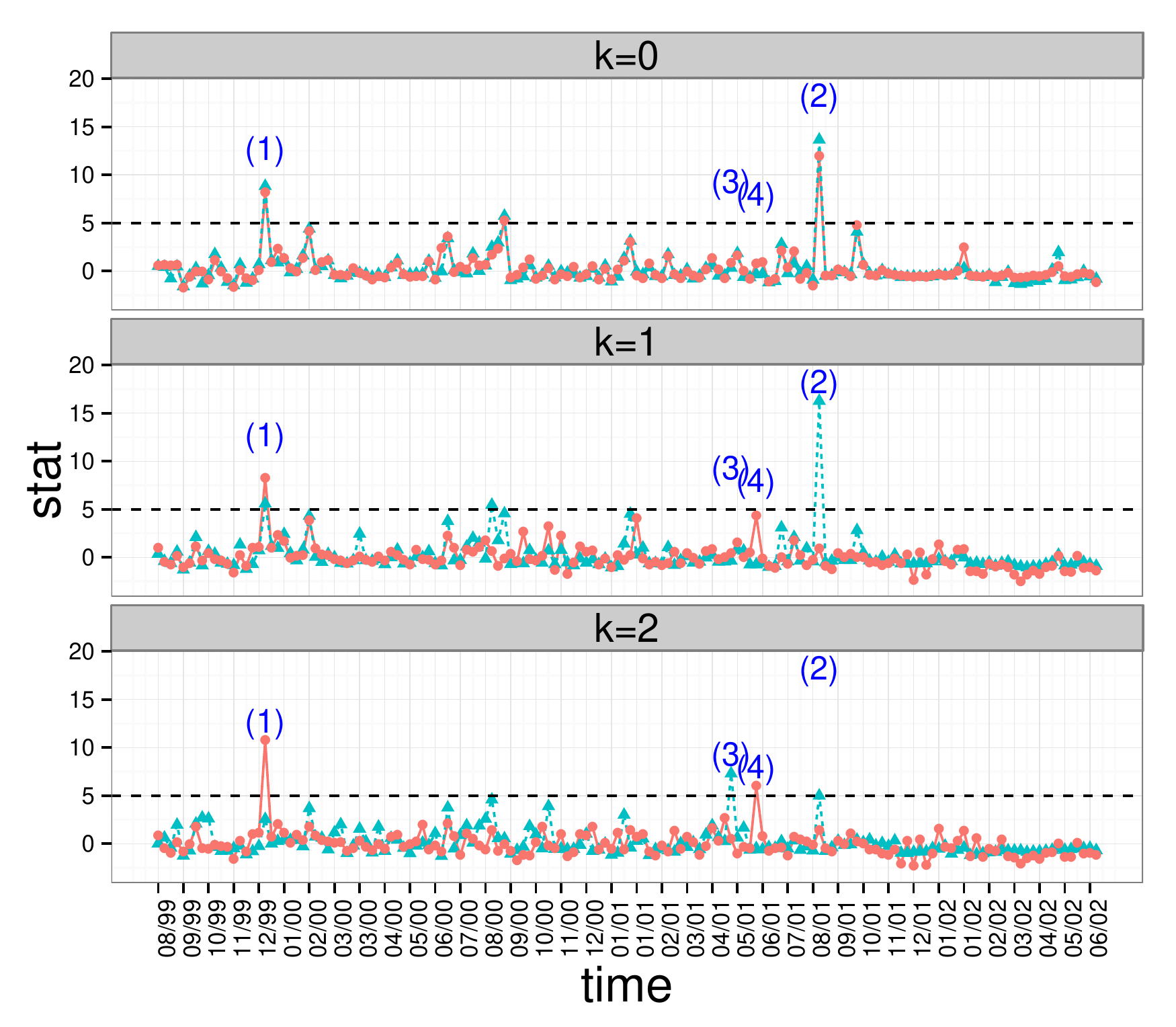}} \\
    \caption{$S_{\tau,\ell,k}(t;\Psi)$(sea green) and
$S_{\tau,\ell,k}(t;\Phi)$(orange), the temporally-normalized
standardized scan statistics using $\tau=\ell=20$, in time series of
Enron email-graphs from August 1999 to June 2002. Top: $k=0$; Middle:
$k=1$; Bottom : $k=2$. In the case $k=0$, both $S_{20,20,0}(t;\Psi)$
and $S_{20,20,0}(t;\Phi)$ show detections
($S_{\tau,\ell,k}(t;\cdot)>5$) at observation mark (1) and (2); in the
case $k=1$, both $S_{20,20,1}(t;\Psi)$ and $S_{20,20,1}(t;\Phi)$ show
detections at observation mark (1), $S_{20,20,1}(t;\Psi)$ also
indicates an anomaly at observation mark (2); in the case $k=2$,
$S_{20,20,2}(t;\Psi)$ detects anomalies at observation mark (2) and
(3) but $S_{20,20,2}(t;\Phi)$ captures anomalies at observation mark
(1) and (4). Detailed analyses on each observation (1)(2)(3)(4) are
provided in \S{\ref{sec:experiment}} respectively.
}\label{fig:enron1}
\end{figure}

% flatex input end: [enron_Heng.tex]

%\input{intro.tex}
% flatex input: [conclusion.tex]
\section{Conclusion \& Discussion}
\label{sec:concl--disc} This paper has summarized a generative latent
position model for time series of graphs and set up the change-point
detection problem in time series of graphs in terms of stochastic
block models.  Then we have proposed the way of dealing with
change-point detection through the use of scan statistics
$S_{\tau,\ell,k}(t;\Psi)$ and $S_{\tau,\ell,k}(t;\Phi)$ constructed
from two different locality statistics $\Psi$ and $\Phi$
respectively. We derived the limiting properties for four
representative instances of locality-based scan statistics
$S_{1,0,0}(t;\Psi),S_{1,0,0}(t;\Phi),S_{1,0,1}(t;\Psi)$ and
$S_{1,0,1}(t;\Phi)$. The limiting properties were then
used to derive estimates for the power of the tests.

The simulation experiments indicate that the analytic power estimates,
even when they are limited in scope, are useful in answering some
important questions about the locality statistics. In particular, it
was shown that $\Psi$ and $\Phi$ are both admissible with respect to
one another when $\tau=1,\ell=0,k=0$.  In addition, if
$\tau=1,\ell=0,k=1$, it is worthwhile to note that $\Psi$, compared
with $\Phi$, is inadmissible but computationally inexpensive. For
instance, in order to complete $\tau$-step vertex-dependent
normalization calculation presented in Eq.(\ref{vertex-norm-eq1}), we
have to record previous $\tau$-step graphs to calculate
$J_{t,t';k}(v)$ if the underlying locality statistic is
$\Phi$. However, if the underlying locality statistic is $\Psi$, graph
storage is not necessary and recording only the previous $\tau$-step
statistics $\Psi_{t';k}(v)$ is sufficient. Furthermore, the power
estimates are also useful for reasoning about the behavior on more
complicated models without $\{G_t\}$ independencies
assumption, such as the latent process model proposed in
\cite{lee11}. \cite{lee11} builds up a latent process
model for time series of attributed graphs based on a random dot
process model. Having $n$ vertices governed by $n$ individual
continuous-time finite-state stochastic processes, this model
generates a time series of dependent attributed random graphs, or
equivalently, conditioning on the sample paths of the stochastic
processes, the graphs are independent. \cite{lee11} also provides two
approximations to the exact latent process model.  The first order
approximation is the stochastic blockmodel which gives rise to a time
series of independent random graphs with independent edges. The second
order approximation corresponds to the random dot product model which
gives rise to a time series of independent random dot product
graphs. Both of these approximations are presented in
\S~\ref{sec:latent-position-model}.

The investigations presented in this paper do not take into account
attributes on the edges. The incorporation of edge attributes into the
current paper is, however, straightforward. For example,
\cite{tang:_attrib} handles attributes by linear fusion, and many of
the results there can be adapted to the current paper. In particular,
one can define fused locality statistics for attributed graphs. Power
estimates for these locality statistics can be derived in a similar
manner to those presented in this paper. Other considerations, e.g.,
optimal fusion parameters, can also be investigated. However, the
statistics considered in \cite{tang:_attrib} are only temporally
normalized and does not contain a vertex dependent
normalization. Thus, the derivation of their limiting properties are
much less involved. In addition, as
the experimental results in Fig~\ref{fig:power-maxd} shown, the vertex
dependent normalization does lead to improved statistical power in
many situations of interest.

Anomaly detection in dynamic graphs has applications in diverse areas,
e.g., predicting the emergence of subgroups within an organization,
monitoring disease spread in public network, detecting modules of
cancer and metastasis communities in Protein-protein interaction (PPI)
network. We envision that these and many other applications will
benefit from the kind of investigations outlined in this
paper. However, much remains to be done, both mathematically and
computationally. We list here some aspects that have not been
(sufficiently) addressed in the paper.

\begin{enumerate}
\item Besides this paper, \cite{tang:_attrib} also investigated
$S_{\tau,\ell,k}(t;\Psi)$ for cases $\tau = 0$, $\ell\to\infty$, and
$k \leq 1$ under the SBM setting. However, power-estimates for other
more complex locality-based scan statistics, such as
$S_{\tau,\ell,k}(t;\cdot)$ for $1<\tau<\infty $, $0<\ell<\infty$ and
$k \geq 2$, remain to be investigated.
\item Locality statistics based on $\Psi$ can be readily computed in a
real-time streaming data environment, in contrast to those based on
$\Phi$. Thus, the adaption or approximation of locality statistics
based on $\Phi$ for streaming environments is of interest.
\item Power-estimates for locality statistics under the random dot
product model setting. The limiting distributions, even for the
simplest locality statistics, are currently unavailable.
\end{enumerate}

% Hypothesis testing on time-series of (attributed) graphs has
% applications in diverse areas, e.g., social network analysis (wherein
% vertices represent individual actors or organizations), connectome
% inference (wherein vertices are neurons or brain regions) and text
% processing (wherein vertices represent authors or documents). The
% current paper had not touched upon many

%%% Local Variables:
%%% mode: latex
%%% TeX-master: "ieee-usthem"
%%% End:

% flatex input end: [conclusion.tex]

%\input{intro.tex}

\ifCLASSOPTIONcompsoc
  \section*{Acknowledgments}
\else
  \section*{Acknowledgment}
\fi

\label{sec:acknowledgements} This work was partially supported by
Johns Hopkins University Human Language Technology Center of
Excellence (JHU HLT COE), and the XDATA program of the Defense
Advanced Research Projects Agency (DARPA) administered through Air
Force Research Laboratory contract FA8750-12-2-0303.

\appendix[Proof of some stated results]
% flatex input: [appendix.tex]
\section{Appendix} In this appendix, we provide proofs of theorems,
propositions and corollaries presented in
\S~\ref{sec:power-estimates-max-d} and
\ref{sec:power-estimates-scan}.\\ \\

\begin{proof}[Theorem \ref{thm:1}] Firstly, we investigate the case
that the underlying locality statistic is $\Psi$. We will derive the
limiting property of $S_{1,0,0}(t;\Psi)$ for $t =
t^{\ast}$ in some detail. The property of
$S_{1,0,0}(t;\Psi)$ when $t<t^\ast$ can be derived in a similar
manner.\\ As $\tau=1$ and $\ell=0$, for any $t$, we have
$\widetilde{\Psi}_{t;1,0}(v)=\Psi_{t;0}(v)-\Psi_{t-1;0}(v)$ from
Eq. (\ref{vertex-norm-eq1}) and Eq. (\ref{vertex-norm-eq2}).  Without
loss of generality, let us assume $v\in [n_i]$ and divide
$\Psi_{t;0}(v)$ into two parts with $t=t^\ast$ and $t=t^\ast-1$:
\begin{equation*}
\Psi_{t^\ast;0}(v)=X_1+X_2
\end{equation*}
where $X_1\sim Bin(n-n_i,p),X_2\sim Bin(n_i-1,\mathbf{P}^A_{i,i})$;
\begin{equation*}
\Psi_{t^\ast-1;0}(v)=X_3+X_4
\end{equation*}
where $X_3\sim Bin(n-n_i,p),X_4\sim Bin(n_i-1,\mathbf{P}^0_{i,i})$.\\
Since $G_{t^\ast-1}$ and $G_{t^\ast}$ are independent, we have
\begin{equation}\label{eq:18}
\begin{split}
 &\dfrac{\widetilde{\Psi}_{t^\ast;1,0}(v)-(n_i-1)(\mathbf{P}^A_{i,i}-\mathbf{P}^0_{i,i})}{\sqrt{np(1-p)}}\\
=&\dfrac{\Psi_{t^\ast;0}(v)-[(n-n_i)p+(n_i-1)\mathbf{P}^A_{i,i}]}{\sqrt{np(1-p)}}\\
 &-\dfrac{\Psi_{t^\ast-1;0}(v)-[(n-n_i)p+(n_i-1)\mathbf{P}^0_{i,i}]}{\sqrt{np(1-p)}}\\
=&\dfrac{X_1-(n-n_i)p}{\sqrt{(n-n_i)p(1-p)}}\cdot\dfrac{\sqrt{(n-n_i)p(1-p)}}{\sqrt{np(1-p)}}\\
 &-\dfrac{X_3-(n-n_i)p}{\sqrt{(n-n_i)p(1-p)}}\cdot\dfrac{\sqrt{(n-n_i)p(1-p)}}{\sqrt{np(1-p)}}\\
 &+\dfrac{X_2-(n_i-1)\mathbf{P}^A_{i,i}}{\sqrt{(n_i-1)\mathbf{P}^A_{i,i}(1-\mathbf{P}^A_{i,i})}}\cdot\dfrac{\sqrt{(n_i-1)\mathbf{P}^A_{i,i}(1-\mathbf{P}^A_{i,i})}}{\sqrt{np(1-p)}}\\
 &-\dfrac{X_4-(n_i-1)\mathbf{P}^0_{i,i}}{\sqrt{(n_i-1)\mathbf{P}^0_{i,i}(1-\mathbf{P}^0_{i,i})}}\cdot\dfrac{\sqrt{(n_i-1)\mathbf{P}^0_{i,i}(1-\mathbf{P}^0_{i,i})}}{\sqrt{np(1-p)}}\\
\xrightarrow{d}&
        ~\mathcal{N}(0,1)\cdot C_1-\mathcal{N}(0,1)\cdot C_2+\mathcal{N}(0,1)\cdot C_3-\mathcal{N}(0,1)\cdot C_4\\
\xrightarrow{d}&
        ~\mathcal{N}(0,C)\\
\end{split}
\end{equation}
where
\begin{align*}
%\shortintertext{where}
C_1&=C_2=\lim_{n\to\infty}\sqrt{\dfrac{n-n_i}{n}},\\
C_3&=\dfrac{\sqrt{(n_i-1)\mathbf{P}^A_{i,i}(1-\mathbf{P}^A_{i,i})}}{\sqrt{np(1-p)}},\\
C_4&=\dfrac{\sqrt{(n_i-1)\mathbf{P}^0_{i,i}(1-\mathbf{P}^0_{i,i})}}{\sqrt{np(1-p)}},\\
C&=\sum_{i=1}^{i=4}C_i^2
\end{align*}
Next, plug in $\mathbf{P}^A_{i,i}$ and $\mathbf{P}^0_{i,i}$ into Eq. (\ref{eq:18}), we obtain
\begin{equation*}
\dfrac{\widetilde{\Psi}_{t^\ast;1,0}(v)-\mathbf{1}_{\{i=B\}}n_B\delta}{\sqrt{Cnp(1-p)}}
            \xrightarrow{d}\mathcal{N}(0,1), ~  v \in [n_i]
\end{equation*}

We can show that the dependency among the
$\{\widetilde{\Psi}_{t^\ast;1,0}(v)\}_{v\in V(G_t)}$ is negligible by
showing that the correlation between any two of the
$\widetilde{\Psi}_{t^\ast;1,0}(v)$ goes to $0$ sufficiently fast as $n
\rightarrow \infty$.  For $u$ and $v$ in block $[n_i]$,
\begin{equation*}
corr(\widetilde{\Psi}_{t^\ast;1,0}(u),\widetilde{\Psi}_{t^\ast;1,0}(v))
%=\dfrac{cov(Z_v,Z_u)}{\sqrt{var(\widetilde{\Psi}_{t^\ast,1,0}(u))}
%\sqrt{var(\widetilde{\Psi}_{t^\ast,1,0}(u))}}
\leq\dfrac{1}{Cnp(1-p)}=O(\dfrac{1}{n})
\end{equation*}
Hence, the sample maximum of $\{Y_v\}_{v\in[n_i]}$
converges to the sample maximum of $n_i$ i.i.d $\mathcal{N}(0,1)$
random variables where
$Y_v=\dfrac{\widetilde{\Psi}_{t^\ast;1,0}(v)-
\mathbf{1}_{\{i=B\}}n_B\delta}{\sqrt{Cnp(1-p)}}$
(\cite{Berman1964}, Theorem $3.1$). Also, it is known
that the sample maximum of i.i.d $\mathcal{N}(0,1)$ random variables
weakly converges to the Gumbel distribution
(\cite{Galambos1987}, \S~$2.3$). One then verifies that
the composition of above weak convergences still holds (see e.g. proof
of Proposition 5 in \cite{tang:_attrib}) and
we thus have
  \begin{equation*}
    \dfrac{W_A(n_i;\Psi)-\mu_A(n_i;\Psi)}{\gamma_A(n_i;\Psi)} \stackrel{d}{\rightarrow}\mathcal{G}(0,1).
  \end{equation*}
Eq.(\ref{eq:8}) and Eq. (\ref{temporal-norm}) then implies that
\begin{equation*}
 S_{1,0,0}(t^\ast;\Psi)=\max_{v\in[n]}\widetilde{\Psi}_{t^\ast;1,0}(v)=
\max_{1\leq i \leq B}\{W_A(n_i;\Psi)\}.
\end{equation*}
That is, the maximum of
$\widetilde{\Psi}_{t^\ast;1,0}(v)$ over all $n$ vertices is equivalent
to the maximum of $W_A(n_i;\Psi)$ over all $B$ blocks where
$W_A(n_i;\Psi)$ converges to $\mathcal{G}(0,1)$ under proper
normalization

Similarly, the case when $t<t^\ast$ can be derived through the same
approaches above. The limiting property of $S_{1,0,0}(t;\Psi)$ with
$t<t^\ast$ then has the form in Eq. (\ref{eq:23}) with variations of
$\mu_0(n_i;\Psi)$ and $\gamma_0(n_i;\Psi)$ for the normalization of
$W_0(n_i;\Psi)$.

We now consider the case where the underlying locality statistic being
$\Phi$. The derivation of limiting property of
$S_{1,0,0}(t;\Phi)$ for $t = t^\ast$ is given below. The derivation of
the limiting property of $S_{1,0,0}(t;\Phi)$ for
$t<t^\ast$ is similar and can be obtained with minor changes.\\ Let's
assume $v\in [n_i]$, from Eq.(\ref{eq:7}) to (\ref{vertex-norm-eq2}),
\begin{equation*}
    \Phi_{t^\ast,t^\ast;0}(v)=X_1+X_2
\end{equation*}
where $X_1\sim Bin(n-n_i,p),X_2\sim Bin(n_i-1,\mathbf{P}^A_{i,i})$ and
\begin{equation*}
   \Phi_{t^\ast,t^\ast-1;0}(v)|G_{t^\ast}=X_3+X_4
\end{equation*} where $X_3\sim Bin(X_1,p),X_4\sim
Bin(X_2,\mathbf{P}^0_{i,i})$.

Because
$\widetilde{\Phi}_{t^\ast;1,0}(v)=\Phi_{t^\ast,t^\ast;0}(v)-\Phi_{t^\ast,t^\ast-1;0}(v)$,
$\widetilde{\Phi}_{t^\ast;1,0}(v)$ counts the number of edges, for
vertex $v$, appearing in $G_{t^\ast}$ but disappearing in
$G_{t^\ast-1}$. Accordingly, the edge is independently counted with
probability $\mathbf{P}^A_{i,i}(1-\mathbf{P}^0_{i,i})$ to neighbors in
$[n_i]$ and $p(1-p)$ to neighbors in $[n]\backslash[n_i]$
respectively. That is,
\begin{equation*}
\widetilde{\Phi}_{t^\ast;1,0}(v)=B_3+B_4
\end{equation*}
where $B_3\sim Bin(n-n_i,p(1-p)), B_4 \sim Bin(n_i-1,\mathbf{P}^A_{i,i}(1-\mathbf{P}^0_{i,i}))$.

By the central limit theorem, we have
\begin{equation} \label{eq:19}
    \begin{split}
    &\dfrac{\widetilde{\Phi}_{t^\ast;1,0}(v)-[(n-n_i)p(1-p)+(n_i-1)\mathbf{P}^A_{i,i}(1-\mathbf{P}^0_{i,i})]}{\sqrt{np(1-p)[1-p(1-p)]}}\\
   =&\dfrac{B_3-(n-n_i)p(1-p)}{\sqrt{(n-n_i)p(1-p)[1-p(1-p)]}}\\
    & \cdot\dfrac{\sqrt{(n-n_i)p(1-p)[1-p(1-p)]}}{\sqrt{np(1-p)[1-p(1-p)]}}\\
    &+\\
    &\dfrac{B_4-(n_i-1)\mathbf{P}^A_{i,i}(1-\mathbf{P}^0_{i,i})}{\sqrt{(n_i-1)\mathbf{P}^A_{i,i}(1-\mathbf{P}^0_{i,i})[1-\mathbf{P}^A_{i,i}(1-\mathbf{P}^0_{i,i})]}}\\
    \cdot &\dfrac{\sqrt{(n_i-1)\mathbf{P}^A_{i,i}(1-\mathbf{P}^0_{i,i})[1-\mathbf{P}^A_{i,i}(1-\mathbf{P}^0_{i,i})]}}{\sqrt{np(1-p)[1-p(1-p)]}}\\
\xrightarrow{d}&
        \mathcal{N}(0,1)\cdot C_1+\mathcal{N}(0,1)\cdot C_2 \\
\xrightarrow{d}&
       ~ \mathcal{N}(0,C)\\
    \end{split}
\end{equation}
where
\begin{align*}
C_1&=\lim_{n\to\infty}\sqrt{\dfrac{n-n_i}{n}},\\
C_2&=\lim_{n\to\infty}\dfrac{\sqrt{(n_i-1)\mathbf{P}^A_{i,i}(1-\mathbf{P}^0_{i,i})[1-\mathbf{P}^A_{i,i}(1-\mathbf{P}^0_{i,i})]}}{\sqrt{np(1-p)[1-p(1-p)]}},\\
C&=\sum_{i=1}^{i=2}C_i^2.
\end{align*}
Similarly, after plugging $\mathbf{P}^0_{i,i}$ and
$\mathbf{P}^A_{i,i}$ into Eq. (\ref{eq:19}),  we obtain
\begin{equation*}
    \begin{split}
&\dfrac{\widetilde{\Phi}_{t^\ast;1,0}(v)-np(1-p)-\xi_0(n_i;\Phi)-\mathbf{1}_{\{i=B\}}n_B\delta(1-p)}{\sqrt{Cnp(1-p)[1-p(1-p)]}}\\
            \xrightarrow{d}&\mathcal{N}(0,1).
    \end{split}
\end{equation*}
For locality statistic $\Phi$, the dependency among
$\{\widetilde{\Phi}_{t^\ast;1,0}(v)\}_{v\in[n]}$ is also negligible
because
\begin{equation*}
\begin{split}
    &corr(\widetilde{\Phi}_{t^\ast;1,0}(u),\widetilde{\Phi}_{t^\ast;1,0}(v))\\
    =&\dfrac{cov(\widetilde{\Phi}_{t^\ast;1,0}(u),\widetilde{\Phi}_{t^\ast;1,0}(v))}{{Cnp(1-p)[1-p(1-p)]}}\\
    \leq&\dfrac{1}{Cn{p(1-p)[1-p(1-p)]}}=O(\dfrac{1}{n})
\end{split}
\end{equation*}

Therefore by following the same
procedures of reasoning the limiting distribution of $W_A(n_i;\Psi)$,
we can also obtain
\begin{equation*}
\dfrac{W_A(n_i;\Phi)-\mu_A(n_i;\Phi)}{\gamma_A(n_i;\Phi)} \stackrel{d}{\rightarrow}\mathcal{G}(0,1)
\end{equation*}
where $W_A(n_i;\Phi)=\max_{v\in[n_i]}\widetilde{\Phi}_{t^\ast;1,0}(v)$.\\
Thus, $S_{1,0,0}(t^\ast;\Phi)$ is the maximum of $W_A(n_i;\Phi)$ over $B$ blocks as desired.
\end{proof}

%%%%%%%%%%%%%%%%%%%%%%%%%%%%%%%%%%%%%%%%%%%%%%% Proof of Corollary 1 %%%%%%%%%%%%%%%%%%%%%%%%%%%%%%%%%%%%%%%%%%%%%%%%%%%%%%
\begin{proof}[Corollary \ref{corollary:max-d}] The limiting
distributions of $\widetilde{\Psi}_{t^\ast-1;1,0}(v)$ and
$\widetilde{\Psi}_{t^\ast;1,0}(v)$ derived in the proof
Theorem~\ref{thm:1} provides that, under $H_0$,
\begin{equation*}
\dfrac{\widetilde{\Psi}_{t^\ast-1;1,0}(v)-0}{\sqrt{Cnp(1-p)}}
            \xrightarrow{d}\mathcal{N}(0,1), ~  v \in [n_i]
\end{equation*}
and, under $H_A$,
\begin{equation*}
\dfrac{\widetilde{\Psi}_{t^\ast;1,0}(v)-\mathbf{1}_{\{i=3\}}n_3\delta}{\sqrt{Cnp(1-p)}}
            \xrightarrow{d}\mathcal{N}(0,1), ~  v \in [n_i]
\end{equation*}

Accordingly, the ratio of the shift in the mean, from null to
alternative, over the standard deviation of $\widetilde{\Psi}_{t;1,0}(v)$
for each vertex would be
$\dfrac{\mathbf{1}_{\{i=3\}}n_3\delta}{\sqrt{Cnp(1-p)}}$. We obtain
two relationships between $\beta_{\Psi}$ and $\alpha$ on the basis of
the order of $n_3$:
\begin{enumerate}
\item
if $n_3=o(\sqrt{n})$, the ratio approaches to $0$ and thus implies $\beta_{\Psi}=\alpha$.
\item if $n_3=\Omega(\sqrt{n})$, then $\exists k>0$ such that
$\dfrac{\mathbf{1}_{\{i=3\}}n_3\delta}{\sqrt{Cnp(1-p)}}\geq k>0$ as
$n\to \infty$ which implies $\beta_{\Psi}>\alpha$.
\end{enumerate}
Likewise, from \textbf{Theorem \ref{thm:1}}, the
limiting distributions of $\widetilde{\Phi}_{t;1,0}(v)$ under null and
alternative respectively are
\begin{equation*}
    \begin{split}
&\dfrac{\widetilde{\Phi}_{t^\ast-1;1,0}(v)-np(1-p)-\xi_0(n_i;\Phi)}{\sqrt{Cnp(1-p)[1-p(1-p)]}}
            \xrightarrow{d}\mathcal{N}(0,1).
    \end{split}
\end{equation*}
\begin{equation*}
    \begin{split}
&\dfrac{\widetilde{\Phi}_{t^\ast;1,0}(v)-np(1-p)-\xi_0(n_i;\Phi)-\mathbf{1}_{\{i=3\}}n_3\delta(1-p)}{\sqrt{Cnp(1-p)[1-p(1-p)]}}\\
\xrightarrow{d}&\mathcal{N}(0,1).
    \end{split}
\end{equation*}
The relationship between $\beta_{\Phi}$ and $\alpha$
is more involved when $\xi_0(n_2;\Phi)$ are included. In order to
clarify the order dominance relationship between $\xi_0(n_2;\Phi)$ and
$n_3\delta(1-p)$, there are five separate cases to be considered:
\begin{enumerate}
\item if $n_3=o(\sqrt{n})$, as $n\to\infty$,
$\widetilde{\Phi}_{t;1,0}(v)$ share the same mean and variance under
both $H_0$ and $H_A$, thus $\beta_{\Phi}=\alpha$.
\item if $n_3=\Theta{(\sqrt{n})} =\Theta(n_2)$,
$\dfrac{n_3\delta(1-p)}{\sqrt{Cnp(1-p)[1-p(1-p)]}}$ and
$\dfrac{\xi_0(n_2;\Phi)}{\sqrt{Cnp(1-p)[1-p(1-p)]}}$ have the same
order $\Theta(1)$ so that the increment
,$\dfrac{n_3\delta(1-p)}{\sqrt{Cnp(1-p)[1-p(1-p)]}}=\Theta(1)$, is not
negligible and implies $\beta_\Phi >\alpha$.
\item if $n_3=\omega{(\sqrt{n})} = \Theta{(n_2)}$, whether $\beta_\Phi
>\alpha $ is determined by if $P(\argmax
\widetilde{\Phi}_{t;1,0}(v)\in [n_3])$ under $H_A$ is larger than
under $H_0$. In fact, if $\dfrac{\xi_0(n_2;\Phi)}{n_3\delta(1-p)}>1$,
$P(\argmax \widetilde{\Phi}_{t;1,0}(v)\in [n_2])=1$ as $n\to\infty$
under both $H_0$ and $H_A$, hence $\beta_\Phi =\alpha$.  Otherwise,
$n_3\delta(1-p)$ in $[n_3]$ contributes to the power increment.
\item if $n_3=\Omega{(\sqrt{n})} =\omega(n_2)$, $n_3\delta(1-p)$
dominates $\xi_0(n_2;\Phi)$ in the limit thereby the location shift in
block $[n_3]$ results in $P(\argmax
\widetilde{\Phi}_{t^\ast;1,0}(v)\in [n_3])=1$ and thus
$\beta_\Phi>\alpha$.
\item if $n_3=\Omega{(\sqrt{n})} = o{(n_2)}$, whether
$n_3\delta(1-p)$ leads to a power increment depends on the sign of
$\xi_0(n_2;\Phi)$. If $h+p<1$ such that $\xi_0(n_2;\Phi)$ being
positive, $P(\argmax \widetilde{\Phi}_{t;1,0}(v)\in [n_2])=1$ under
both $H_0$ and $H_A$ as $n\to\infty$ because $n_3=o(n_2)$. On the
contrary, if $h+p\geq 1$, $\xi_0(n_2;\Phi)<0$ enables $P(\argmax
\widetilde{\Phi}_{t;1,0}(v)\in [n_3])$ to increase from $H_0$ to
$H_A$. Thus, we have
$
  \beta_\Phi =\alpha \,\, \text{if $h+p<1$};
 \beta_\Phi >\alpha  \,\, \text{if $h+p\geq 1$}.
$
\end{enumerate}
\end{proof}

%%%%%%%%%%%%%%%%%%%%%%%%%%%%%%%%%%%%%%%%%%%%%%% Proof of dist of 3-blocks %%%%%%%%%%%%%%%%%%%%%%%%%%%%%%%%%%%%%%%%%%%%%%%%%%%%%%
\begin{proof}[Proposition \ref{3block-dist}] We present a sketch of
the proof based on arguments from \cite{andrey:12} for the case where
the underlying locality statistic is $\Psi$. The case where the
underlying locality statistic is $\Phi$ follows from the proof of
Theorem~\ref{thm:2}.

Let $v\in[n_i](i\in\{1,2,3\})$, locality statistics
$\Psi_{t^\ast,t^\ast;1}(v)$ and $\Psi_{t^\ast,t^\ast-1;1}(v)$ are
respectively decomposed as follows
\begin{equation*}\label{us-scan-decomposition1}
    \Psi_{t^\ast,t^\ast;1}(v)=X_i+\sum_{j\neq i}X_j+\sum_{j=1}^{3}Y_j+\sum_{1\leq j<k\leq 3}^{3}Z_{jk}, ~ v\in[n_i]
\end{equation*}
where
\begin{align*}
X_i&\sim Bin(n_i-1,\mathbf{P}^A_{i,i}),\\
X_j&\sim Bin(n_j,\mathbf{P}^A_{i,j}),\\
Y_j|X_j&\sim Bin({X_j\choose2},\mathbf{P}^A_{j,j}),\\
Z_{jk}|X_j,X_k&\sim  Bin(X_jX_k,\mathbf{P}^A_{j,k}).
\end{align*}
and
\begin{equation*} \label{us-scan-decomposition2}
\Psi_{t^\ast,t^\ast-1;1}(v)=X'_i+\sum_{j\neq
i}X'_j+\sum_{j=1}^{3}Y'_j+\sum_{1\leq j<k\leq 3}^{3}Z'_{jk}, ~
v\in[n_i]
\end{equation*}
where
\begin{align*}
X'_i&\sim Bin(n_i-1,\mathbf{P}^0_{i,i}),\\
X'_j&\sim Bin(n_j,\mathbf{P}^0_{i,j}),\\
Y'_j|X'_j&\sim Bin({X'_j\choose2},\mathbf{P}^0_{j,j}),\\
Z'_{jk}|X'_j,X'_k&\sim  Bin(X'_jX'_k,\mathbf{P}^0_{j,k}).
\end{align*}
Hence, when $\mathbf{P}^0$ and $\mathbf{P}^A$ are substituted, we have
\begin{align*}
&\widetilde{\Psi}_{t^\ast;1,1}(v)=\Psi_{t^\ast,t^\ast;1}(v)-\Psi_{t^\ast,t^\ast-1;1}(v)\\
\begin{split} =&[(X_i+\sum_{j\neq i}X_j)-(X'_i+\sum_{j\neq i}X'_j)]\\
&+[(\sum_{j=1}^{3}Y_j+\sum_{1\leq j<k\leq
3}^{3}Z_{jk})-(\sum_{j=1}^{3}Y'_j+\sum_{1\leq j<k\leq
3}^{3}Z'_{jk})]\\ =&[(X_i+\sum_{j\neq i}X_j)-(X'_i+\sum_{j\neq
i}X'_j)]\cdot\\ &\Big[1+\dfrac{p}{2}[(X'_i+\sum_{j\neq
i}X'_j)+(X_i+\sum_{j\neq i}X_j)]\Big]\\
&+\dfrac{h-p}{2}(X_2^2-{X'_2}^2)+\dfrac{\delta}{2}X^2_3\\
=&\widetilde{\Psi}_{t^\ast;1,0}(v)\cdot
\Big[1+\dfrac{p}{2}[(X'_i+\sum_{j\neq i}X'_j)+(X_i+\sum_{j\neq
i}X_j)]\Big]\\ &+\dfrac{h-p}{2}(X_2^2-{X'_2}^2)+\dfrac{\delta}{2}X^2_3
\end{split}
\end{align*}

Thus, by using similar approaches given in the proof of
lemma $3.2$ and lemma $3.3$ from \cite{andrey:12}, we obtain, as $n\to
\infty$,
\begin{equation*} \arg max_{v\in
[n_i]}\widetilde{\Psi}_{t^\ast;1,0}(v)=\arg max_{v\in
[n_i]}\widetilde{\Psi}_{t^\ast;1,1}(v)
\end{equation*}
and
\begin{align*}
\begin{split}
&\lim P(W'_A(n_i;\Psi)>\mu'_A(n_i;\Psi))\\
=&\lim P(W_A(n_i;\Psi)>\mu_A(n_i;\Psi))
\end{split}
\end{align*}
where $W'_A(n_i;\Psi)=\max_{v\in
[n_i]}\widetilde{\Psi}_{t^\ast;1,1}(v)$ and $W_A(n_i;\Psi)=\max_{v\in
[n_i]}\widetilde{\Psi}_{t^\ast;1,0}(v)$.  This leads to the fact that
$(W'_A(n_i;\Psi)-\mu'_A(n_i;\Psi))/\beta'_A(n_i;\Psi)$ follows
standard Gumbel distribution $\mathcal{G}(0,1)$ and
$S_{1,0,1}(t^\ast;\Phi) =\max_{1\leq i\leq 3}W'_A(n_i;\Phi)$. Similar
arguments apply to $S_{1,0,1}(t^\ast-1;\Psi)$.
\end{proof}

%%%%%%%%%%%%%%%%%%%%%%%%%%%%%%%%%%%%%%%%%%%%%%% Proof of power of 3-blocks %%%%%%%%%%%%%%%%%%%%%%%%%%%%%%%%%%%%%%%%%%%%%%%%%%%%%%

%\begin{proof}[Proposition \ref{3block-power}]
%The conclusion trivially follows from the distributions given in \textbf{Proposition \ref{3block-dist}}.
%\end{proof}
%\quad \\

Before proving Theorem~\ref{thm:2}, we state and prove a technical
lemma on the correlations among the $\{\widetilde{\Phi}_{t;1,1}(v)\}$.
\begin{lemma}\label{lemma} Let $G_{t-1}$ and $G_{t}$ be two
independent Erd\"{o}s-R\'{e}nyi graphs with connectivity probability
$p$, i.e., $G_{t-1} \sim G(n,p)$ and $G_{t} \sim G(n,p)$. For each
$v$, $\widetilde{\Phi}_{t;1,1}(v)$ is defined according to
Eq. (\ref{vertex-norm-eq1}). Then for any pair of vertices $u$ and
$v$, the correlation between $\widetilde{\Phi}_{t;1,1}(u)$ and
$\widetilde{\Phi}_{t;1,1}(v)$ is of order $O(\tfrac{1}{n})$ for $n
\rightarrow \infty$.
\end{lemma}
\begin{proof}
From Eq. (\ref{vertex-norm-eq1}) and (\ref{vertex-norm-eq2}),
for any pair of vertices $(u,v)$,
\begin{equation} \label{TotalCov}
\begin{split}
&cov(\widetilde{\Phi}_{t;1,1}(u),\widetilde{\Phi}_{t;1,1}(v))\\
=&cov({\Phi}_{t,t;1}(u),{\Phi}_{t,t;1}(v))-cov({\Phi}_{t,t;1}(u),{\Phi}_{t,t-1;1}(v))-\\
 &cov({\Phi}_{t,t-1;1}(u),{\Phi}_{t,t;1}(v))+cov({\Phi}_{t,t-1;1}(u),{\Phi}_{t,t-1;1}(v))
\end{split}
\end{equation}
We then consider to decompose
${\Phi}_{t,t;1}(u)$ into two parts representing the cardinalities of
two disjoint sets of edges.
\begin{equation*}
\begin{split}
{\Phi}_{t,t;1}(u)=X_t(u)+Y_t(u)
\end{split}
\end{equation*}
where the intuitive interpretations behind two terms are listed below:
\begin{align*} X_t(u)=&|\{(u,w):(u,w)\in E(G_{t}) \text{ and } w \in
N_1(u;G_t)\backslash\{u\} \}|\\ Y_t(u)=&|\{(w_1,w_2): (w_1,w_2) \in
E(G_{t}), w_1<w_2 \text{ and } \\ & w_1,w_2 \in
N_1(u;G_t)\backslash\{u\}\}|
\end{align*}
Also, ${\Phi}_{t,t-1;1}(u)$ is decomposed into two terms as well.
\begin{equation*}
\begin{split}
{\Phi}_{t,t-1;1}(u)=X_{t-1}(u)+Y_{t-1}(u)
\end{split}
\end{equation*}
where the intuitive interpretations behind two terms are listed below:
\begin{align*} X_{t-1}(u)=&|\{(u,w):(u,w)\in E(G_{t})\cap E(G_{t-1})\\
& \text{ and } w \in N_1(u;G_t)\backslash\{u\}|\\
Y_{t-1}(u)=&|\{(w_1,w_2): (w_1,w_2) \in E(G_{t-1}), w_1<w_2 \text{ and
} \\ & w_1,w_2 \in N_1(u;G_{t})\backslash\{u\}\}|
\end{align*}

Similarly, ${\Phi}_{t,t;1}(v)$ and ${\Phi}_{t,t-1;1}(v)$ are
decomposed with the same structure.  By expanding above decompositions
into Eq. (\ref{TotalCov}), we have the following table recording $16$
terms and their signs in (\ref{TotalCov}).\\
\begin{table}[tbp]
  \centering
  \caption{Decomposition of the covariance terms in
    $S_{\tau,l,k}(\cdot, \Psi)$ for $\tau = 1, l = 0, k = 1$.}
  \label{tab:covariance-decomposition}
\begin{tabular}{|c|c|c|c|c|} \hline
$cov(\cdot,\cdot)$ & $X_t(v)$
& $X_{t-1}(v)$ & $Y_{t}(v)$ & $Y_{t-1}(v)$ \\ $X_t(u)$ & $+$ & $-$ &
$\color{blue}+$ & $\color{blue}-$ \\ $X_{t-1}(u)$ & $-$ & $+$ &
$\color{green}-$ & $\color{green}+$ \\ $Y_t(u)$ &
$\color{blue}+^{\dagger}$ & $\color{green}-^{\ddagger}$ &
$\color{magenta}+^{\amalg}$ & $\color{magenta}-$ \\ $Y_{t-1}(u)$ &
$\color{blue}-^{\dagger}$ & $\color{green}+^{\ddagger}$ &
$\color{magenta}-^{\amalg}$ & $\color{magenta}+$ \\ \hline
\end{tabular}
\end{table}
In Table~\ref{tab:covariance-decomposition}, all terms
earning the same color ({\color{blue}blue}, {\color{green} green} or
{\color{magenta}magenta}) and same positive/negative sign are
symmetric. Additionally, the terms having the same mark
($\dagger,\ddagger$ or $\amalg$) are canceled out due to the fact
$Y_{t}(\cdot)|X_{t}(\cdot) \stackrel{iid}{\sim}
Y_{t-1}(\cdot)|X_{t}(\cdot)$. More concretely, for example, for four
terms marked by {\color{blue}blue}, we have
$cov(X_t(u),Y_t(v))=cov(Y_t(u),X_t(v))=cov(Y_{t-1}(u),X_t(v))=cov(X_t(u),Y_{t-1}(v))$. The
first and third equality are guaranteed by symmetry property. The
second equality holds because $Y_t(u)$ and $Y_{t-1}(u)$ share the same
conditional distribution, $Bin({X_{t}(u)\choose 2},p)$, given
$X_t(u)$. That is,
$cov(Y_t(u),X_t(v)|X_t(u))=cov(Y_{t-1}(u),X_t(v)|X_t(u))$ and hence
$cov(Y_t(u),X_t(v))=cov(Y_{t-1}(u),X_t(v))$ with application of law of
total covariance.\\ We now return to Eq. (\ref{TotalCov}). The above
reasoning gives
\begin{equation*}
\begin{split}
&cov(\widetilde{\Phi}_{t;1,1}(u),\widetilde{\Phi}_{t;1,1}(v))\\
=&cov(X_{t}(u),X_{t}(v))-cov(X_{t}(u),X_{t-1}(v))-\\
 &cov(X_{t-1}(u),X_{t}(v))+cov(X_{t-1}(u),X _{t-1}(v))\\
=&O(n).
\end{split}
\end{equation*}

The last equality holds because the Cauchy-Schwarz
inequality guarantees each of fours term are $O(n)$ where
$X_{t}(\cdot)\sim Bin(n-1,p)$ and $X_{t-1}(\cdot)\sim Bin(n-1,p^2)$.\\

In the following, to compute $var(\widetilde{\Phi}_{t;1,1}(u))$,
$\widetilde{\Phi}_{t;1,1}(u)$ is decomposed as
\begin{align*}
\widetilde{\Phi}_{t;1,1}(u)=X_{t}+Y_{t}-X_{t-1}-Y_{t-1}
\end{align*}
where
\begin{align*}
X_{t}&\sim Bin(n-1,p),\\
Y_{t}|X_{t}&\sim Bin({X_{t}\choose 2},p),\\
X_{t-1}|X_{t}&\sim Bin(X_{t},p),\\
Y_{t-1}|X_{t}&\sim Bin({X_{t}\choose 2},p),\\
Y_{t}|X_{t} &\perp Y_{t-1}|X_{t}.
\end{align*}
By applying law of total variance, we reach the following variance order estimation
\begin{align*}
\begin{split}
&var(\widetilde{\Phi}_{t;1,1}(u))\\
=&\Theta(var(Y_{t}-Y_{t-1}))\\
=& \Theta(E[var(Y_{t}-Y_{t-1}|X_{t})]+var[E(Y_{t}-Y_{t-1}|X_{t})])\\
=& \Theta(E[2{X_{t}\choose 2}p(1-p)]+var[0])\\
=& \Theta(n^2p^3(1-p))
\end{split}
\end{align*}
Therefore, it follows that
\begin{equation*}
corr(\widetilde{\Phi}_{t;1,1}(u),\widetilde{\Phi}_{t;1,1}(v))= O(\dfrac{1}{n})
\end{equation*}
as desired.
\end{proof}
%%%%%%%%%%%%%%%%%%%%%%%%%%%%%%%%%%%%%%%%%%%%%%% Proof of Thm of scan-B Blocks %%%%%%%%%%%%%%%%%%%%%%%%%%%%%%%%%%%%%%%%%%%%%%%%%%%%%%
\begin{proof}[Theorem \ref{thm:2}] Again, to avoid redundant
arguments, we only provide derivations of limiting distribution of
$\widetilde{\Phi}_{t^\ast;1,1}(v)$ and the case $t<t^\ast$ can be
achieved in the same approach. Let $v\in[n_i]$, locality statistics
$\Phi_{t^\ast,t^\ast;1}(v)$ and $\Phi_{t^\ast,t^\ast-1;1}(v)$ are
respectively decomposed as follows:
\begin{equation}\label{them-scan-decomposition1}
    \Phi_{t^\ast,t^\ast;1}(v)=X_i+\sum_{j\neq i}X_j+\sum_{j=1}^{B}Y_j+\sum_{1\leq j<k\leq B}^{B}Z_{jk}, ~ v\in[n_i]
\end{equation}
where
\begin{align*}
X_i&\sim Bin(n_i-1,\mathbf{P}^A_{i,i}),\\
X_j&\sim Bin(n_j,\mathbf{P}^A_{i,j}),\\
Y_j|X_j&\sim Bin({X_j\choose2},\mathbf{P}^A_{j,j}),\\
Z_{jk}|X_j,X_k&\sim  Bin(X_jX_k,\mathbf{P}^A_{j,k}).
\end{align*}
and
\begin{equation} \label{them-scan-decomposition2}
    \Phi_{t^\ast,t^\ast-1;1}(v)=X'_i+\sum_{j\neq i}X'_j+\sum_{j=1}^{B}Y'_j+\sum_{1\leq j<k\leq B}^{B}Z'_{jk}, ~ v\in[n_i]
\end{equation}
where
\begin{align*}
X'_i|X_i&\sim Bin(X_i,\mathbf{P}^0_{i,i}),\\
X'_j|X_j&\sim Bin(X_j,\mathbf{P}^0_{i,j}),\\
Y'_j|X_j&\sim Bin({X_j\choose2},\mathbf{P}^0_{j,j}),\\
Z_{jk}|X_j,X_k&\sim  Bin(X_jX_k,\mathbf{P}^0_{j,k}).
\end{align*}
Accordingly, the mean of $\widetilde{\Phi}_{t^\ast;1,1}(v)$ is estimated as follows
\begin{align*}
\begin{split}
&E(\widetilde{\Phi}_{t^\ast;1,1}(v))\\
=&E(\Phi_{t^\ast,t^\ast;1}(v)-\Phi_{t^\ast,t^\ast-1;1}(v))\\
=&E(X_i+\sum_{j\neq i}X_j-X'_i-\sum_{j\neq i}X'_j)+\\
 &E(\sum_{j=1}^{B}Y_j+\sum_{1\leq j<k\leq B}^{B}Z_{jk}-\sum_{j=1}^{B}Y'_j-\sum_{1\leq j<k\leq B}^{B}Z'_{jk})\\
=&E(\widetilde{\Phi}_{t^\ast;1,0}(v))+ \\
 &E(\sum_{j=1}^{B}Y_j+\sum_{1\leq j<k\leq B}^{B}Z_{jk}-\sum_{j=1}^{B}Y'_j-\sum_{1\leq j<k\leq B}^{B}Z'_{jk})\\
=&E(\widetilde{\Phi}_{t^\ast;1,0}(v))+E(Y_B-Y'_B)+o(n) \\
=&np(1-p)+\xi_0(n_i;\Phi)+\bm{1}_{\{i=B\}}n_B\delta(1-p)+\zeta(n_i,p,\delta,i)\\
 &+o(n).
\end{split}
\end{align*} Under our setting of $\mathbf{P}^0$ and $\mathbf{P}^A$,
the penultimate equality is obtained easily because $Z_{jk}$ and
$Z'_{jk}$ share the same distribution and $Y_j$ share the same
distribution with $Y'_j$ except $j=B$. \\ Now let's consider the
estimation of $var(\widetilde{\Phi}_{t^\ast;1,1}(v))$ since the exact
derivation of $var(\widetilde{\Phi}_{t^\ast;1,1}(v))$, through the use
of law of total variance, is tedious. Due to the assumption
$[n_1,n_2,\dots,n_B]=[\Theta(n),o(n),\dots,o(n)]$ and decompositions
in Eq.(\ref{them-scan-decomposition1}) and
Eq.(\ref{them-scan-decomposition2}), instead we express variance of
$\widetilde{\Phi}_{t^\ast;1,1}(v)$ as
\begin{align*}
\begin{split}
&var(\widetilde{\Phi}_{t^\ast;1,1}(v))\\
=&var(\Phi_{t^\ast,t^\ast;1}(v)-\Phi_{t^\ast,t^\ast-1;1}(v))\\
=&var(Y_1-Y'_1)+O(n^{2-\epsilon})\\
=&Cn^2p^3(1-p)+O(n^{2-\epsilon})
\end{split}
\end{align*}
Thus, the central limit theorem leads to
\begin{equation*}
    \begin{split}
&\dfrac{\widetilde{\Phi}_{t^\ast;1,1}(v)-E(\widetilde{\Phi}_{t^\ast;1,1}(v))}{\sqrt{Cn^2p^3(1-p)}}
\xrightarrow{d} ~ \mathcal{N}(0,1)
    \end{split}
\end{equation*}
According to Lemma~\ref{lemma}, dependencies among
$\{\widetilde{\Phi}_{t^\ast;1,1}(v)\}_{v\in[n_i]}$ are negligible and
thus
\begin{equation*}
  \begin{split}
\dfrac{max_{v\in[n_i]}\widetilde{\Phi}_{t^\ast;1,1}(v)-
\mu'_A(n_i;\Phi)}{\gamma'_A(n_i;\Phi)}
& = \dfrac{W'_A(n_i;\Phi)-\mu'_A(n_i;\Phi)}{\gamma'_A(n_i;\Phi)}
\\ &
\overset{\mathrm{d}}{\longrightarrow} \mathcal{G}(0,1).
\end{split}
\end{equation*}

% thus %
%$\dfrac{W'_A(n_i;\Phi)-\mu'_A(n_i;\Phi)}{\gamma'_A(n_i;\Phi)}= %
%\dfrac{max_{v\in[n_i]}\widetilde{\Phi}_{t^\ast;1,1}(v)-
%\mu'_A(n_i;\Phi)}{\gamma'_A(n_i;\Phi)}$
% is a standard Gumbel random variable.
Through similar arguments as
in Theorem~\ref{thm:1}, we can show that
$S_{1,0,1}(t^\ast;\Phi)=\max_{1\leq i\leq B}W'_A(n_i;\Phi)$.
\end{proof}

%%%%%%%%%%%%%%%%%%%%%%%%%%%%%%%%%%%%%%%%%%%%%%% Proof of power of 3-blocks %%%%%%%%%%%%%%%%%%%%%%%%%%%%%%%%%%%%%%%%%%%%%%%%%%%%%%
\begin{proof}[Corollary \ref{corollary:scan}] This corollary is a
generalization of Proposition \ref{3block-dist} and Proposition
\ref{3block-power}. The underlying idea is as follows. In the model
presented at the beginnning of \S~\ref{sec:power-estimates-max-d}, the
variation of number of chatty blocks before $t^\ast-1$ makes no
difference on the sensitivity of statistics $S_{1,0,1}(t;\Psi)$ and
$S_{1,0,1}(t;\Phi)$ as long as the orders of chatty blocks are
$o(n)$. Namely, in the limiting case, $\beta'_{\Phi}$ and
$\beta'_{\Psi}$ are functions of $n_B$ and independent of
$\{n_2,n_3,\dots n_{B-1}\}$.  We can then extend the power comparison
conclusion from Proposition \ref{3block-power} for $B=3$ to the
general case.  The details are somewhat tedious and are omitted.
\end{proof}

%%% Local Variables:
%%% mode: latex
%%% TeX-master: "ieee-usthem"
%%% End:

% flatex input end: [appendix.tex]


\ifCLASSOPTIONcaptionsoff
  \newpage
\fi

\bibliography{IEEEabrv,usthem}
\end{document}